\documentclass[10pt]{iopart}

\usepackage{iopams,amsfonts,amssymb,amsthm} 

\pdfoutput=1
\usepackage{etex}
\usepackage{cite}
\usepackage{latexsym}
\usepackage{rawfonts}
\usepackage[usenames,dvipsnames,svgnames]{xcolor}
\usepackage{bbm}
\usepackage{latexsym}
\usepackage{multirow}
\usepackage{rotating}
\usepackage{lscape}
\usepackage{graphicx}
\usepackage{xcolor}
\usepackage{fancybox}
\usepackage{ifmtarg}
\usepackage{fancyhdr}
\usepackage{wrapfig}
\usepackage{hyperref}
\usepackage{tikz}

\input prepictex
\input pictex
\input postpictex

\def\qd{\quad}

\def\2;{\;\;}

\def\eps{\epsilon}

\def\NatNu{{\mathbb N}}

\def\RealN{{\mathbb R}}

\def\Ref#1{(\ref{#1})}

\def\C#1{{\mathcal #1}}

\def\h#1{{\mathbf{\widehat{\hbox{$#1$}}}}}

\def\hi{\mathbf{\widehat{\hbox{\textbf{\textit{i}}}}}}
\def\hj{\mathbf{\widehat{\hbox{\textbf{\textit{j}}}}}}

\def\h#1{\mathbf{\widehat{\hbox{\textbf{\textit{#1}}}}}}

\def\Sfrac#1#2{\hbox{\large $\frac{#1}{#2}$}}
\def\sfrac#1#2{\hbox{\nor $\frac{#1}{#2}$}}

\def\LB{\left(}         \def\RB{\right)}
        
\def\LA{\left\langle}        \def\RA{\right\rangle}

\def\lfl{\!\left\lfloor} \def\rfl{\right\rfloor\!}
       
\def\LH{\left[}        \def\RH{\right]}

\newcommand{\Claim}{\vspace{2mm}\noindent{{\scshape \nor Claim:} }} 
\newcommand{\ProofClaim}{{\scshape \noindent \nor Proof of claim: }}
\newcommand{\cqed}{\hskip 2mm $\triangle$ \vskip 2mm}

\def\nor{\normalsize}
\def\fns{\scriptsize}

\def\svv{{\;\hbox{$|$}\;}}

\def\edge#1#2{{\langle #1 \hspace{0.85pt}{\sim}\hspace{0.85pt} #2 \rangle}}

\def\plus{{\hspace{0.85pt}{+}\hspace{0.85pt}}}
\def\minus{{\hspace{0.85pt}{-}\hspace{0.85pt}}}

\def\tiny#1{\scalebox{0.75}{#1}}

\definecolor{blue}{rgb}{0,0.18,0.39}
\definecolor{RoyalBlue}{rgb}{0,0.2,0.7}

\definecolor{Maroon}{cmyk}{0,0.87,0.68,0.62}
\definecolor{Brown}{rgb}{0.7,0.3,0}
\definecolor{Navy}{rgb}{0.3,0.0,0.4}
\definecolor{Red}{cmyk}{0,1,1,0}
\definecolor{BrickRed}{cmyk}{0.16,0.89,0.61,0.02}
\definecolor{DarkRed}{cmyk}{0,1,1,0.5}
\definecolor{DarkBlue}{cmyk}{1,1,0,0.2}
\definecolor{DarkGreen}{cmyk}{1,0,1,0.4}
\definecolor{Green}{cmyk}{1,0,1,0}
\definecolor{DarkBrown}{cmyk}{0,0.81,1,0.6}
\definecolor{OrangeRed}{cmyk}{0,1,0.87,0}
\definecolor{RedOrange}{cmyk}{0,0.77,0.87,0}
\definecolor{Orange}{cmyk}{0,0.61,0.87,0}
\definecolor{Offwhite}{rgb}{.8,0.9,.8}
\definecolor{Offwhite2}{cmyk}{.04,.02,.01,0}
\definecolor{Tan}{rgb}{0.82,0.70,0.55}
\definecolor{Blue}{rgb}{0,0,1}
\definecolor{RoyalBlue}{rgb}{0.25,0.41,0.88}
\definecolor{Sepia}{rgb}{0.37,0.14,0.07}
\definecolor{myblue}{cmyk}{0.025,0.05,0,0}
\definecolor{Mahogany}{cmyk}{0.18,0.87,1,0.08}

\definecolor{green1}{cmyk}{0.25,0,0.76,0}
\definecolor{green2}{cmyk}{0.25,0,0.76,0.07}
\definecolor{green3}{cmyk}{0.25,0,0.76,0.20}
\definecolor{green4}{cmyk}{0.25,0,0.75,0.30}
\definecolor{green5}{cmyk}{0.25,0,0.75,0.40}
\definecolor{green6}{cmyk}{0.25,0,0.75,0.50}

\definecolor{B02}{cmyk}{0,0.14,0.22,0.12}
\definecolor{B03}{cmyk}{0,0.16,0.26,0.16}
\definecolor{B04}{cmyk}{0,0.19,0.28,0.19}
\definecolor{B05}{cmyk}{0,0.25,0.32,0.25}
\definecolor{B06}{cmyk}{0,0.31,0.36,0.31}
\definecolor{B07}{cmyk}{0,0.37,0.40,0.37}
\definecolor{B08}{cmyk}{0,0.46,0.46,0.46}
\definecolor{B09}{cmyk}{0,0.55,0.52,0.54}
\definecolor{B10}{cmyk}{0,0.69,0.61,0.62}
\definecolor{B11}{cmyk}{0,0.78,0.70,0.68}
\definecolor{B12}{cmyk}{0,0.93,0.85,0.60}
\definecolor{B13}{cmyk}{0.25,1,0.6,0.50}
\definecolor{B14}{cmyk}{0.5,1,0.30,0.40}
\definecolor{B15}{cmyk}{0.75,1,0,0.30}

\definecolor{C02}{cmyk}{0,0.22,0.14,0.12}
\definecolor{C03}{cmyk}{0,0.26,0.16,0.16}
\definecolor{C04}{cmyk}{0,0.28,0.19,0.19}
\definecolor{C05}{cmyk}{0,0.32,0.25,0.25}
\definecolor{C06}{cmyk}{0,0.36,0.31,0.31}
\definecolor{C07}{cmyk}{0,0.40,0.37,0.37}
\definecolor{C08}{cmyk}{0,0.46,0.46,0.46}
\definecolor{C09}{cmyk}{0,0.52,0.55,0.54}
\definecolor{C10}{cmyk}{0,0.61,0.69,0.62}
\definecolor{C11}{cmyk}{0,0.70,0.78,0.68}
\definecolor{C12}{cmyk}{0,0.85,0.93,0.60}
\definecolor{C13}{cmyk}{0.25,0.60,1,0.50}
\definecolor{C14}{cmyk}{0.5,0.30,1,0.40}
\definecolor{C15}{cmyk}{0.75,0,1,0.30}

\newtheorem{theorem}{Theorem}
\newtheorem{lemma}{Lemma}

\begin{document}

\title[Supermultiplicative relations of lattice clusters]{Supermultiplicative relations in models 
of interacting self-avoiding walks and polygons.}

\author{EJ Janse van Rensburg$^1$}

\address{\sf$^1$Department of Mathematics and Statistics, 
York University, Toronto, Ontario M3J~1P3, Canada\\}
\ead{rensburg@yorku.ca}
\vspace{10pt}
\begin{indented}
\item[]\today
\end{indented}

\begin{abstract}
Fekete's lemma shows the existence of limits in subadditive sequences.
This lemma, and generalisations of it, also have been used to prove the 
existence of thermodynamic limits in statistical mechanics.  In this paper it
is shown that the two variable supermultiplicative relation 
\[ p_{n_1}(m_1)\,p_{n_2}(m_2) \leq p_{n_1+n_2}(m_1\plus m_2) \]
together with mild assumptions, imply the existence of the limit
\[ \log \C{P}_\#(\eps) = \lim_{n\to\infty} \sfrac{1}{n} \log p_n(\lfloor \eps n \rfloor).\]
This is a generalisation of Fekete's lemma.   The existence of $\log \C{P}_\#(\eps)$
is proven for models of adsorbing walks and polygons, and for pulled polygons.
In addition, numerical data are presented estimating the general shape
of $\log \C{P}_\# (\eps)$ of models of square lattice self-avoiding walks
and polygons.
\end{abstract}

%
\vspace{0pc}

\submitto{\JPA}

%
%


 
\section{Introduction}

Let $c_n$ be the number of self-avoiding walks of length $n$ steps from
the origin in a regular lattice.  In the square lattice, 
$c_0=1$, $c_1=4$, $c_2=12$, and so on.  A walk of length 
$n\plus m$ steps can be cut into two subwalks in the 
$(n\plus 1)^{th}$ vertex, the first subwalk of length $n$, and the 
second of length $m$.  This shows that $c_n$ is a submultiplicative function:
\begin{equation}
c_n\, c_m \geq c_{n+m} .
\end{equation}
Taking logarithms shows that $\log c_n$ is subadditive and by Fekete's lemma
\cite{HF57} the limit
\begin{equation}
\kappa_d = \lim_{n\to\infty} \Sfrac{1}{n} \log c_n
= \inf_{n\geq 0} \Sfrac{1}{n} \log c_n
\label{2}
\end{equation}
exists \cite{BH57}.  The constant $\kappa_d$ is the 
\textit{connective constant} of the self-avoiding walk.  In the hypercubic 
lattice $d^n \leq c_n \leq (2d)^n$ so that $\log d \leq \kappa_d \leq \log (2d)$.
The \textit{growth constant} $\mu_d$ of the self-avoiding walk 
is defined by $\log \mu_d = \kappa_d$, and the result in equation \Ref{2} 
shows that $c_n \geq \mu_d^n$ and $c_n = \mu_d^{n+o(n)}$.
See references \cite{H57,H60,K63,K64} for more classical results on $c_n$.

The existence of a growth constant in other models of lattice clusters
usually follows from a supermultiplicative relation.  For example, this
is the case for \textit{lattice polygons} \cite{H61A}.
Let $p_n$ be the number of lattice polygons in the hypercubic lattice
where $n$ is even (the hypercubic lattices are bipartite, so that
$p_n=0$ if $n$ is odd).  Two polygons can be concatenated as 
illustrated in figure \ref{figure1} and this shows that
\begin{equation}
p_n\,p_m \leq C_d\, p_{n+m},
\end{equation}
for even values of $n$ and $m$, where $C_d$ is a constant dependent 
on the dimension $d$.  In the square lattice $C_d=1$ and in the cubic lattice 
$C_d=2$.  This supermultiplicative relation implies that $- \log \LB p_n/C_d \RB$ 
is a subadditive function on even numbers, and it again follows 
from Fekete's lemma that
\begin{equation}
\lambda_d = \lim_{n\to\infty} \Sfrac{1}{2n} \log p_{2n}
= \sup_{n\geq 0} \Sfrac{1}{2n} \log p_{2n}
\label{4}
\end{equation}
exists. It is also proven that $\lambda_d = \kappa_d = \log \mu_d$  \cite{HW62A}, 
and it follows from equation \ref{4} that $p_{2n} \leq C_d\, \mu_d^{2n}$ and 
$p_{2n} = C_d\, \mu_d^{2n+o(n)}$. 

\begin{figure}[t!]
\beginpicture
\setcoordinatesystem units <1pt,1pt> 
\setlinear

\setplotarea x from -100 to 250, y from -20 to 40

\setcoordinatesystem units <1pt,1pt> point at -70 20

\put
{\beginpicture
\setplotsymbol ({$\cdot$})
\plot 0 10  -10 10  -10 0  -20 0  -20 10  -30 10  -30 0  -50 0 /
\plot -50 0  -50 20  -40 20  -40 30  -30 30  -30 40  -20 40  -20 30 /
\plot -20 30  -10 30  -10 20  0 20  0 10 /
\plot 10 10  10 30 10 60 40 60 
50 60 50 50 20 50 20 40 40 40 /
\plot 40 40 40 30 40 0  30 0  30 10  10 10 /
\multiput {\fns$\bullet$} at
0 10  -10 10  -10 0  -20 0  -20 10  -30 10  -30 0  -40 0 -50 0 -50 10  -50 20  
-40 20  -40 30  -30 30 -30 40  -20 40  -20 30  -10 30  -10 20  0 20  0 10 
10 10  10 20 10 30 10 40 10 50 10 60 20 60 30 60 40 60 
50 60 50 50 40 50 30 50 20 50 20 40 30 40 40 40 40 30 40 20 40 10 
40 0  30 0  30 10  20 10 10 10 /

\setplotsymbol ({.})

\plot 4 20  6 20 /
\plot 4 10  6 10 / 
\put {$p_n$} at -25 20 
\put {$p_m$} at 25 25 
\endpicture} at 0 30 

\endpicture
\label{figure1}
\caption{\textit{Concatenating two lattice polygons in the square lattice
by placing them with the right-most top-most edge of the first polygon
parallel to the left-most bottom-most edge of the second polygon.  
Deleting these two edges, and replacing them with edges 
shown gives a polygon of length $n\plus m$.  Since there are $p_n$ 
choices of the polygon on the left, and $p_m$ choices on the 
right, $p_n\, p_m \leq p_{n+m}$ in the square lattice.}}
\end{figure}

More generally, constructions involving lattice clusters similar to that 
in figure \ref{figure1} may give rise to other supermultiplicative 
relations, making it necessary to consider generalisations of Fekete's lemma.
The first is due to Wilker and Whittington \cite{WW79} and involves 
subadditive relations of the kind
\begin{equation}
a_{n+f(m)} \leq a_n + a_m,
\quad\hbox{where 
$\displaystyle \lim_{m\to\infty} \Sfrac{1}{m} f(m) = 1$}.
\end{equation}
A second generalisation is due to Hammersley \cite{H62} and it
is concerned with subadditive relations of the kind
\begin{equation}
a_{n+m} \leq a_n + a_m +\phi_{n+m},
\quad\hbox{where 
$\displaystyle \sum_{n=N_0}^\infty \Sfrac{\phi_n}{n(n+1)} < \infty$}.
\label{6}
\end{equation}

While Fekete's lemma is explicitly applicable to subadditive sequences,
it also applies to superadditive sequences, and in this context a generalisation
of it to functions of two variables will be considered in this paper.
In particular, multivariate functions taking $\NatNu \times \NatNu$ into the
non-negative reals $\RealN^+$ and satisfying a supermultiplicative relation 
\begin{equation}
p_{n_1}^\#(m_1)\, p_{n_2}^\#(m_2) \leq p_{n_1+n_2}^\# (m_1\plus m_2) 
\label{7} 
\end{equation}
will be considered here, so that $\log p_n^\#(m)$ satisfies a superadditive
relation.  It is shown in section 
\ref{section2}, given certain assumptions, that the limit
\begin{equation}
\log \C{P}_\#(\epsilon) 
= \lim_{n\to\infty} \Sfrac{1}{n} \log p_n^\#(\lfloor \epsilon n \rfloor)
\label{8}
\end{equation}
exists, and it is a concave function of $\epsilon$.  

In the statistical mechanics of models of lattice clusters, relations of the type 
in equation \Ref{7} arise in constructions similar to the concatenation in
figure \ref{figure1} (see references \cite{G99,HTW82,JvR15,MSW88}).  In section
\ref{section3} the limit in equation \Ref{8} is generalised:  Suppose that
$\delta_n$ is a positive real-valued function on $\NatNu$ such
that $\lim_{n\to\infty} \sfrac{1}{n} \delta_n = \delta$.  Then
\begin{equation}
\log \C{P}_\#(\delta) 
= \lim_{n\to\infty} \Sfrac{1}{n} \log p_n^\#(\delta_n)
\label{9A}
\end{equation} 
and the limit on the right hand side exists. In section \ref{free energy} the relation of 
$\log P_\#(\eps)$ to the free energy of the underlying statistical mechanics
model is examined, and in section \ref{examples} some models in the square
lattice, including adsorbing walks, adsorbing and pulled polygons are
presented.

\section{Existence of the limit in equation \Ref{8}}
\label{section2}

Let $n\in\NatNu$ and $m\in\NatNu$.  Let $p_n^\#(m)$
be a function on $\NatNu^2$ with range in the non-negative
reals. Define $p_n^\# = \sum_m p_n^\#(m)$  and make the
following assumptions for $p_n^\# (m)$.

\vspace{3mm}
\noindent{\bf Assumptions:} 
{\parindent=24pt
\begin{itemize}
\item[\bf (1)] There exist a finite constant $C>0$, and functions
$A_n$ and $B_n$, such that $0 \leq A_n \leq B_n \leq Cn$ and
$p_n^\# (m) >0$ if $A_n \leq m \leq B_n$, and $p_n ^\# (m) = 0$
otherwise.  Moreover, assume that 
$\liminf_{n\to\infty} \sfrac{1}{n}A_n < \limsup_{n\to\infty} \sfrac{1}{n}B_n$.
\item[\bf (2)] There exists a constant $K>0$ such that for $n\in\NatNu$,
$0\leq p_n ^\# (m) \leq K^n$ for $m\in \NatNu$.
\item[\bf (3)] The function $p_n ^\# (m)$ satisfies a supermultiplicative
inequality of the type 
\begin{equation}
p_{n_1} ^\# (m_1)\, p_{n_2} ^\# (m_2) \leq
   p_{n_1+n_2} ^\# (m_1\plus m_2) . 
\label{9}
\end{equation}
\end{itemize} \qed
}
\smallskip

Define 
\begin{equation}
\eps_m = \liminf_{n\to\infty} \sfrac{1}{n}A_n,\qd\hbox{and}\qd
\eps_M = \limsup_{n\to\infty} \sfrac{1}{n}B_n ,
\label{10}
\end{equation}
so that $\eps_m< \eps_M$ by assumption (1).  Since $C$ is finite, 
it follows that $0 \leq \eps_m < \eps_M \leq C < \infty$.

\begin{lemma}
\label{lemma1}
For each $n\in \NatNu$ it is the case that 
$\eps_m = \lim_{n\to\infty} \sfrac{1}{n} A_n 
\leq \sfrac{1}{n} A_n$ and 
$\eps_M = \lim_{n\to\infty} \sfrac{1}{n} B_n 
\geq \sfrac{1}{n} B_n$.
\end{lemma}

\begin{proof} By assumptions (1) and (3),
\[ 0 < p^\#_n(A_n)\, p^\#_m(A_m) \leq p^\#_{n+m}(A_n\plus A_m). \]
This shows that $A_n \plus A_m \geq A_{n+m}$.  Thus $A_n$ is 
a subadditive function on $\NatNu$.  By Fekete's lemma
\[ \eps_m 
= \lim_{n\to\infty} \sfrac{1}{n} A_n = \inf_{n\in \NatNu}\sfrac{1}{n} A_n .\]

Similarly
\[ 0 < p^\#_n(B_n)\, p^\#_m(B_m) \leq p^\#_{n+m}(B_n\plus B_m). \]
This shows that $B_n\plus B_m \leq B_{n+m}$.  Hence $B_n$ is
a superadditive function on $\NatNu$.  By Fekete's lemma
\[ \eps_M 
= \lim_{n\to\infty} \sfrac{1}{n} B_n = \sup_{n\in \NatNu}\sfrac{1}{n} B_n .\]
This completes the proof of the lemma. 
\end{proof}
\smallskip

The function $p_n^\#$ grows exponentially with $n$.

\begin{theorem}
The growth constant $\mu_\#$ of $p_n^\#$ is defined by
\[ \log \mu_\#  =  \lim_{n\to\infty} \sfrac{1}{n} \log p_n^\# > 0 .\]
By the assumption (2), $\mu_\# < K$ and so is finite.
\end{theorem}

\begin{proof}
Sum the left-hand side of equation \Ref{9} over all $\{m_1,m_2\}$:
\[ p_{n_1}^\#\, p_{n_2}^\# 
\leq (B_{n_1}+B_{n_2})\, p_{n_1+n_2}^\# 
\leq B_{n_1+n_2}\, p_{n_1+n_2}^\# .\]
Take logarithms to see that
\[ - \log p_{n_1+n_2}^\# \leq
- \log p_{n_1}^\# - \log p_{n_2}^\# + \log B_{n_1+n_2} . \]
Since $B_n \leq C\,n$ it follows by equation \Ref{6}
that the limit exists as claimed \cite{H62}.  Moreover, by assumption (2)
$\log \mu_\# \leq \log K$. 
\end{proof}
\smallskip

Fixed $\eps \in (\eps_m,\eps_M)$ and choose $m_1 = \lfl \eps n_1\rfl$ 
and $m_1\plus m_2 = \lfl \eps (n_1 \plus n_2) \rfl$ in equation \Ref{9}.  
This gives
\begin{equation}
p_{n_1} ^\# (\lfl \eps n_1\rfl)\, 
p_{n_2} ^\# (\lfl \eps (n_1 \plus n_2) \rfl - \lfl \eps n_1\rfl) 
\leq   p_{n_1+n_2} ^\# (\lfl \eps (n_1 \plus n_2) \rfl) . 
\end{equation}

\begin{lemma}
\label{lemma3.3}
There exists a function $z_{n_1,n_2}$ of $(\eps, n_1,n_2)$
such that 
\[p^\#_{n_1} (\lfl \eps n_1 \rfl)\,
p_{n_2}^\# (\lfl \eps n_2 \rfl+ z_{n_1,n_2}) 
\leq p_{n_1+n_2}^\#(\lfl \eps(n_1\plus n_2) \rfl ) ,\]
and where $|z_{n_1,n_2}| \leq 1$. 
\end{lemma}

\begin{proof} Fix $\eps\in (\eps_m,\eps_M)$, and fix $n_1$ and $n_2$.
If $\lfl \eps n_1 \rfl \plus \lfl \eps n_2 \rfl = \lfl \eps (n_1\plus n_2) \rfl-1$
then 
\[ p^\#_{n_1} (\lfl \eps n_1 \rfl)\, 
   p_{n_2}^\# (\lfl \eps n_2 \rfl \plus 1 ) 
\leq p_{n_1+n_2}^\#(\lfl\eps(n_1\plus n_2) \rfl ). \]  
Otherwise, 
$\lfl \eps n_1 \rfl \plus \lfl \eps n_2 \rfl = \lfl \eps (n_1\plus n_2) \rfl$,
and so 
\[ p^\#_{n_1} (\lfl \eps n_1 \rfl)\, p_{n_2}^\# (\lfl \eps n_2 \rfl)
\leq p_{n_1+n_2}^\#(\lfl \eps(n_1\plus n_2) \rfl ). \]  

Define $z_{n_1,n_2} = 1$ if 
$\lfl \eps n_1 \rfl \plus\lfl \eps n_2 \rfl = \lfl \eps (n_1\plus n_2) \rfl-1$, 
and zero otherwise.  Then these two last inequalities can be combined into
\[ p^\#_{n_1} (\lfl \eps n_1 \rfl)\, 
   p_{n_2}^\# (\lfl \eps n_2 \rfl + z_{n_1,n_2}) 
   \leq p_{n_1+n_2}^\#(\lfl \eps(n_1\plus n_2) \rfl ) \] 
for any $\eps\in(\eps_m,\eps_M)$ and for any $n_1$ and $n_2$. 
\end{proof}
\smallskip

Lemma \ref{lemma3.3} can be used to prove the
existence of the limit in equation \Ref{8}.

\begin{theorem}
\label{theorem5}
Let $\eps\in (\eps_m,\eps_M)$. Then the limit
\[ \log {\cal P}_\# (\epsilon) = \lim_{n\to\infty} \sfrac{1}{n} \log p^\# _n
(\lfloor \epsilon n \rfloor ) \]
exists.   Moreover, for each value of $n$
\[p_n^\#(\lfl\eps n\rfl\! \plus \eta_n) \leq [{\cal P}_\# (\eps)]^n \]
where $\eta_n\in \{0,1\}$ is a function of $n$.
\end{theorem}

\begin{proof}
The proof is similar to the proof of Fekete's lemma; see for example
theorem A.1 in reference \cite{JvR15}.  Fix $\eps\in(\eps_m,\eps_M)$
and let $n=n_1\plus n_2$ be large.

Choose $n_1=n\minus k$ and $n_2=k$ in lemma \ref{lemma3.3}.  
Then the supermultiplicative inequality becomes 
\begin{equation}
p^\#_n (\lfl \eps n \rfl)
\geq p^\#_{n-k} (\lfl \eps (n\minus k) \rfl)\,
   p^\#_k(\lfl \eps k \rfl + z_{n-k,k}) . 
\label{13}
\end{equation}
Put $n=Nm\plus r$, where $m$ is a large fixed natural number,
and with $N_0 \leq r < N_0 \plus m$, for some fixed large $N_0$.
Assume that $n\gg N_0$.
  
Choose $k=r$ so equation \Ref{13} becomes
\[ p^\#_{Nm+r} (\lfl \eps (Nm\plus r) \rfl)
\geq  p^\#_{Nm}(\lfl\eps Nm\rfl)\,
 p^\#_r(\lfl \eps r \rfl + z_{Nm,r}) .\]
Increase $N_0$, if necessary, until $\lfl \eps r \rfl \geq A_r \plus z_{Nm,r}$
(with $A_r$ defined in assumption (1)). Since $\eps > \eps_m$, 
and $z_{Nm,r} \in \{0,1\}$, this is always possible.

Next, consider $p^\#_{Nm}(\lfl\eps Nm \rfl)$ and apply 
equation \Ref{13}, with $k=m$, recursively, for $n=Nm$, 
$(N\minus 1)m$, $\ldots$,$m$. This gives
\begin{eqnarray*}
p^\#_{Nm} (\lfl \eps Nm \rfl ) 
& \geq  p^\#_{(N\minus 1)m} (\lfl \eps((N\minus 1)m) \rfl)
 \, p^\#_{m} (\lfl \eps m  \rfl\!  \plus z_{(N\minus 1)m,m}) \\
& \geq  p^\#_{(N\minus 2)m} (\lfl \eps((N\minus 2)m) \rfl)
 \, \prod_{j=1}^2 \LH
    p^\#_{m} (\lfl \eps m  \rfl\!  \plus z_{(N\minus j)m,m})\RH  \\
& \geq   \cdots \\
& \geq  
    \prod_{j=1}^N \LH p^\#_m (\lfl \eps m\rfl\!  \plus z_{(N-j)m,m} )\RH . 
\end{eqnarray*}
That is, 
\begin{eqnarray*} 
p^\#_n (\lfl \eps n \rfl )
&\geq  p^\#_{Nm}(\lfl\eps Nm\rfl)\, p^\#_r(\lfl \eps r \rfl + z_{Nm,r})
   \nonumber\\
&\geq \prod_{j=1}^N \LH p^\#_m (\lfl \eps m\rfl\!  \plus z_{(N-j)m,m} )\RH 
\, p^\#_r(\lfl \eps r \rfl + z_{Nm,r}) . 
\end{eqnarray*}
Choose $\eta_m \in \{0,1\}$ to be that value of $\ell$ which minimizes
$\min_\ell \{ p^\#_m(\lfl \eps m\rfl\! \plus \ell) \}$ (recall that $m$ is fixed).  
Then $z_{(N-j)m,m}$ in the above may be replaced by $\eta_m$.  
This gives
\[ p^\#_n (\lfl \eps n \rfl )
\geq
\LH p^\#_m (\lfl \eps m\rfl\!  \plus \eta_m) \RH^N\,
 p^\#_r (\lfl \eps r \rfl\!  \plus z_{mN,r}) .\]
Take logarithms of this, and divide by $n=mN\plus r$.  Keep $m$ fixed
and take the limit inferior as $n\to\infty$ on the left hand side.
Since $n=Nm\plus r$ with $N_0 \leq r <N_0\plus m$, $N\to\infty$, and
it follows that
\begin{equation}
\liminf_{n\to\infty} \sfrac{1}{n} \log p^\#_n (\lfl \eps n \rfl) \geq
\sfrac{1}{m} \log p^\#_m (\lfl \eps m \rfl \!  \plus \eta_m) . 
\label{14}
\end{equation}
Now take the limsup as $m\to\infty$ on the right hand side.  Then
\[ \liminf_{n\to\infty} 
\sfrac{1}{n} \log p^\#_n (\lfl \eps n \rfl) \geq
\limsup_{m\to\infty} 
\sfrac{1}{m} \log p^\#_m (\lfl \eps m \rfl\!   \plus \eta_m) . \]

Proceed by choosing $n_1=m\minus k$, $n_2=k$, so that
$n_1\plus n_2=m$ in equation \Ref{9}.  Put $m_1 = \lfl \eps(m\minus k) \rfl$
and $m_1\plus m_2=\lfl \eps m \rfl + \eta_m$ (with $\eta_m$ as
defined above).  Then $m_2=\lfl \eps m \rfl + \eta_m - \lfl \eps(m\minus k) \rfl$.
This gives
\[ p_{m-k}^\# (\lfl \eps(m\minus k) \rfl)\,
 p_k^\# (\lfl \eps m \rfl\! \plus \eta_m \minus \lfl\!  \eps(m\minus k) \rfl ) 
\leq p_m^\# (\lfl \eps m \rfl\!  \plus \eta_m ) . \]
Define $\delta_{m,k}$ by $\lfl \eps k \rfl\!  \plus \delta_{m,k} 
= \lfl \eps m \rfl \minus  \lfl \eps(m\minus k) \rfl$.  
Then $|\delta_{m,k} | \leq 1$, and it follows that 
\[ p_{m-k}^\# (\lfl \eps(m\minus k) \rfl)\,
 p_k^\# (\lfl \eps k  \rfl\!  \plus \delta_{m,k} \plus \eta_m) 
\leq p_m^\# (\lfl \eps m \rfl\!  \plus \eta_m ) . \]
Increase $k$, if necessary, until $\lfl \eps k  \rfl \geq A_k \plus 2 
 \geq A_k \plus \delta_{m,k} \plus \eta_m$, with $A_k$ as defined in 
assumption (1).  Since $\eps > \eps_m$ this is always
possible.  Fix $k$, take logarithms and divide the above by $m$,
and take the limsup on the left hand side as $m\to\infty$.  This gives
\begin{equation}
\limsup_{m\to\infty} \sfrac{1}{m} \log p_m^\# (\lfl \eps m \rfl)
\leq \limsup_{m\to\infty} \sfrac{1}{m} \log p_m^\# (\lfl \eps m \rfl \plus \eta_m ) . 
\label{15}
\end{equation}
Comparison with equation \Ref{14} establishes the existence of 
the limit 
\[\C{P}_\# (\eps) 
= \lim_{n\to\infty} \sfrac{1}{n} \log p_n^\# (\lfl \eps n \rfl) .\] 
By equation \Ref{14} there exists an $\eta_n \in \{0,1\}$ 
such that $p_n^\#(\lfl\eps n\rfl \plus \eta_n) \leq [{\cal P}_\# (\eps)]^n$
for each value of $n$.
\end{proof}
\smallskip

A useful consequence of the proof above is that if $\eta_n\in\{0,1\}$ 
is defined as that value of $\ell\in\{0,1\}$ minimizing
$p_n^\#(\lfl \eps n \rfl\plus\ell)$, then equation \Ref{15} gives
\begin{equation}
\log \C{P}_\#(\eps) 
= \lim_{n\to\infty} \sfrac{1}{n} \log p_n^\# (\lfl \eps n \rfl)
\leq \limsup_{n\to\infty} 
\sfrac{1}{n} \log p_n^\# (\lfl \eps n \rfl \plus \eta_n ) 
\label{13A}
\end{equation}
for $\eps\in(\eps_m,\eps_M)$.

It follows from the supermultiplicative inequality in equation \Ref{9}
that $\log \C{P}_\#(\eps)$ is a concave function of $\eps\in(\eps_m,\eps_M)$.

\begin{theorem}
\label{theorem3.5}
$\log {\cal P}_\# (\eps)$ is a concave function of $\eps\in (\eps_m,\eps_M)$.  
${\cal P}_\# (\eps)$ is continuous in $(\eps_m,\eps_M)$, has 
semi-continuous right- and left-derivatives everywhere in 
$(\eps_m,\eps_M)$, and is differentiable almost everywhere (except at
countably many isolated points) in $(\eps_m,\eps_M)$.
\end{theorem}

\begin{proof}
Choose $\eps$ and $\delta$ so that $\eps_m < \eps < \delta < \eps_M$.

Choose $n_1=n_2=n$, $m_1 = \lfl \eps n \rfl$ and $m_1\plus m_2 = 
\lfl (\eps\plus \delta) n \rfl$ in equation \Ref{9}.  Then 
$m_2 = \lfl (\eps\plus \delta) n \rfl \minus \lfl \eps n \rfl$.  There exists 
an integer $\zeta_n\in\{0,1\}$, dependent on $\delta$ and $\eps$, such that
$m_2 = \lfl \delta n \rfl\! \plus \zeta_n$.  These choices in equation 
\Ref{9} give
\[ p_n^\# (\lfl \eps n \rfl )\, p_n^\# (\lfl \delta n \rfl\!  \plus \zeta_n) \leq
p_{2n}^\# (\lfl (\eps\plus\delta) n \rfl) . \]   

Define $\eta_n$ as before (it is that value of $\ell\in\{0,1\}$ minimizing
$p_n^\# (\lfl \delta n \rfl\!  \plus \ell)$):
\[ p_n^\#(\lfl \delta n\rfl \plus \eta_n) = \min_{\ell\in\{0,1\}}
\{ p_n^\#(\lfl \delta n\rfl \plus \ell) \}
\leq p_n^\#(\lfl\delta n \rfl + \zeta_n) . \]
Replace the factor $p_n^\# (\lfl \delta n \rfl \plus \zeta_n)$
by $p_n^\#(\lfl \delta n \rfl \plus \eta_n)$ to obtain
\[ p_n^\# (\lfl \eps n \rfl )\; p_n^\# (\lfl \delta n \rfl\!  \plus \eta_n) \leq
p_{2n}^\# (\lfl (\eps\plus\delta) n \rfl) . \]
Dividing this by $n$, and taking the limsup 
on the left hand side as $n\to\infty$,  gives
\[ \log {\cal P}_\# (\eps) +
 \limsup_{n\to\infty} \sfrac{1}{n} \log p_n^\# (\lfl \delta n \rfl \plus \eta_n) 
\leq 2 \log {\cal P}_\# (\sfrac{1}{2} (\eps\plus\delta)) . \]
By equation \Ref{13A} it follows that
\[ \log {\cal P}_\# (\eps) + \log {\cal P}_\# (\delta) \leq
2\, \log {\cal P}_\# (\sfrac{1}{2} (\eps\plus\delta)) . \]
Thus, $\log {\cal P}_\# (\eps)$ is concave on $(\eps_m,\eps_M)$
and so differentiable almost everywhere in $(\eps_m,\eps_M)$
(except on a countable set of isolated points).
This shows that the right- and left-derivatives exist everywhere in
$(\eps_m,\eps_M)$ and are semi-continuous.
\end{proof}
\smallskip

\section{Further results}
\label{section3}

\subsection{Existence of the limit in equation \Ref{9A}}

\begin{theorem}
\label{theorem3.6}
Let $\delta_n$ be a sequence of integers such that $A_n < \delta_n < B_n$ 
for all $n\geq M_0$ (where $M_0\in\NatNu$ is fixed).  If
$\lim_{n\to\infty} \sfrac{1}{n}\, \delta_n = \delta \in (\eps_m,\eps_M)$,
then 
\[ \lim_{n\to\infty} \sfrac{1}{n} \log p_n^\# (\delta_n) = \log \C{P}_\# (\delta) . \]
\end{theorem}

\begin{proof}
Let $n\geq M_0$ and apply inequality \Ref{9} $N\minus 1$ times to obtain
\[ \LH p_n^\#(\delta_n) \RH^N \leq p_{nN}^\#(N\delta_n) 
=  p_{nN}^\#(\lfloor nN\,\sfrac{1}{n}\delta_n \rfloor) . \]
Take the logarithm on both sides, divide by $nN$ and let $N\to\infty$.  
Then by theorem \ref{theorem5}, 
\[
\sfrac{1}{n} \log p_n^\#(\delta_n) \leq \log \C{P}_\# (\sfrac{1}{n} \delta_n) . \]
Take the $\limsup$ on the left hand side of this inequality; since 
$\log \C{P}_\# (\eps)$ is a concave function, it is continuous on 
$(\eps_m,\eps_M)$, and in particular since 
$\sfrac{1}{n} \delta_n \to \delta$, the result is
\begin{equation}
\limsup_{n\to\infty} 
\sfrac{1}{n} \log p_n^\#(\delta_n) \leq \log \C{P}_\# (\delta) . 
\label{eqn3.3}
\end{equation}

If $\delta \in (\eps_m,\eps_M)$, then 
\begin{equation}
\delta_n - \lfl \delta (n\minus k) \rfl = \lfl \delta k \rfl +f_n 
\label{eqn3.4}
\end{equation}
where $f_n = o(n)$.  Define $\kappa_n = \lfl \max\{\sqrt{n\,|f_n|},\sqrt{n}\} \rfl$
so that $\kappa_n = o(n)$ and $f_n = o(\kappa_n)$.  Notice that $\kappa_n\to\infty$
as $n\to\infty$.

\Claim There exists an $N$ such that for all $n>\max\{N,M_0\}$,
\[ \delta_{n} - \lfl \delta({n}\minus \kappa_{n}) \rfl > A_{\kappa_n} .\]

\ProofClaim  Let $\kappa_n$ and $f_n$ be as defined above.  

By lemma \ref{lemma1} for all $n\in \NatNu$,
\[ \sfrac{1}{n} A_n \geq \lim_{n\to\infty} \sfrac{1}{n} A_{n} = \eps_m . \]
Let $\alpha = \sfrac{1}{3} (\delta - \eps_m) > 0$.  Since $\kappa_n \to\infty$
as $n\to\infty$, this shows, in particular, that there is an $N_0$ such that
\[ \eps_m \leq  \Sfrac{1}{\kappa_n} A_{\kappa_n} < \eps_m + \alpha,
\qd\hbox{for all $n>N_0$} . \]

By equation \Ref{eqn3.4},
\[ \Sfrac{1}{\kappa_n} \LB \delta_n - \lfl \delta(n \minus \kappa_n) \rfl \RB
= \Sfrac{1}{\kappa_n} \LB \lfl \delta \kappa_n \rfl + f_n \RB . \]
If $n\to\infty$, then $\kappa_n \to\infty$ and this gives
\[ \lim_{n\to\infty}  \Sfrac{1}{\kappa_n} \LB \delta_n 
- \lfl \delta(n \minus \kappa_n) \rfl \RB = \delta  \]
since $f_n = o(\kappa_n)$.
In other words, for $\alpha$ as above, there exists an $N_1$ such that for
all $n> N_1$,
\[ \Sfrac{1}{\kappa_n} \LB \delta_n - \lfl \delta(n \minus \kappa_n) \rfl \RB 
> \delta - \alpha . \]

Put $N= \max\{ N_0,N_1 \}$ and suppose that $n>N$.  Then, by the above
\[  \Sfrac{1}{\kappa_n} A_{\kappa_n} < \eps_m + \alpha
< \delta - \alpha <  \Sfrac{1}{\kappa_n} \LB \delta_n 
- \lfl \delta(n \minus \kappa_n) \rfl \RB .\]

That is, for all $n>N$,
\[ \delta_n - \lfl \delta(n \minus \kappa_n) \rfl > A_{\kappa_n} . \]
This completes the proof of the claim. \cqed

Notice that similar to the proof of the claim above, and by 
lemma \ref{lemma1}, there is an $M$ such that
\[ \delta_{n} - \lfl \delta(n\minus \kappa_n) \rfl < B_{\kappa_n} \]
for all $n > M$.  Indeed, notice that with $\kappa_n$ as defined above,
\[ \lim_{n\to\infty}
 \sfrac{1}{\kappa_n} \LB \delta_{n} - \lfl \delta(n\minus \kappa_n) \rfl \RB
= \delta < \eps_M  = \lim_{n\to\infty} \sfrac{1}{n} B_n. \]
That is, for small $\alpha>0$ and $\alpha < \sfrac{1}{3}(\eps_M-\delta)$, 
there is a $M$ such that for all $n>\max\{M,M_0\}$,
\[ \delta_{n} - \lfl \delta(n\minus \kappa_n) \rfl 
< (\delta \plus \alpha) \kappa_n
< (\eps_M \minus \alpha) \kappa_n < B_{\kappa_n}. \]

Put $J = \max\{N,M,M_0\}$.  Then for all $n>J$, by the claim and the 
last inequality, 
\[ A_{k_n} < \delta_{n} - \lfl \delta(n\minus \kappa_n) \rfl < B_{\kappa_n} . \]
Since $\kappa_n = o(n)$ in the claim above, $n-\kappa_n \to\infty$ as $n\to\infty$.
Use equation \Ref{9} with $n_1=n-\kappa_n$, $n_2=\kappa_n$, 
$m_1=\lfl \delta(n\minus \kappa_n) \rfl$ and 
$m_2=\delta_{n} - \lfl \delta(n\minus \kappa_n) \rfl$.  Then
\[ p_{n-\kappa_n}^\#(\lfl \delta(n\minus \kappa_n) \rfl)\;
p^\#_{\kappa_n}(\delta_{n} - \lfl \delta(n\minus \kappa_n) \rfl) \leq
p^\#_{n}(\delta_{n}) . \]
Take logarithms, divide by $n$, and take the limit inferior on the right hand
side as $n\to\infty$.   Since 
$A_{k_n} < \delta_{n} - \lfl \delta(n\minus \kappa_n) \rfl < B_{\kappa_n}$
for all $n>J$, and $\kappa_n = o(n)$, this gives
\[ \log \C{P}_\#(\delta) = \lim_{n\to\infty}\sfrac{1}{n} \log
p_{n-\kappa_n}^\#(\lfl \delta(n-\kappa_n) \rfl) \leq
\liminf_{n\to\infty} \sfrac{1}{n} \log p_n^\# (\delta_n) \]
by theorem \ref{theorem5}. This completes the proof, by 
equation \Ref{eqn3.3}.
\end{proof}
\smallskip

A corollary of this theorem is that the concavity of $\log \C{P}_\#(\eps)$
can be extended to the closed interval $[\eps_m,\eps_M]$.  This
follows from equation \Ref{9}.  Choose $n_1=n_2=n$,
$m_1=A_n$, and $m_2=\lfl\eps n \rfl$ where $\eps\in(\eps_m,\eps_M)$.
This gives, after taking logarithms and dividing by $n$, 
\begin{equation}
\Sfrac{1}{n} \log p_n^\#(A_n)
+\Sfrac{1}{n} \log p_n^\# (\lfl \eps n \rfl) \leq
  \Sfrac{1}{n} \log p_{2n}^\#(A_n+\lfl \eps n \rfl) .
\end{equation}
By lemma \ref{lemma1}, $A_n \geq n\, \eps_m$ and 
$\lim_{n\to\infty} \sfrac{1}{n} A_n = \eps_m$.  Take the 
limsup on the left hand side above as $n\to\infty$ to obtain
\begin{equation}
\limsup_{n\to\infty} \Sfrac{1}{n} \log p_n^\#(A_n)
+ \log\C{P}_\#(\eps) \leq 2\log \C{P}_\#(\sfrac{1}{2}(\eps_m\plus \eps))
\end{equation}
by theorems \ref{theorem5} and \ref{theorem3.6}.  
Taking $\eps\to\eps_m^+$ gives 
$\limsup_{n\to\infty}\sfrac{1}{n} \log p_n^\#(A_n) 
\leq \log\C{P}_\#(\eps_m^+)$. \textit{Defining} 
$\log \C{P}_\#(\eps_m) = \limsup_{n\to\infty}\sfrac{1}{n} \log p_n^\#(A_n)$
extends $\log\C{P}_\#(\eps)$ to a function which is concave on
to $[\eps_m,\eps_M)$.  

A similar argument shows that for any $\eps\in(\eps_m,\eps_M)$,
\begin{equation}
\log \C{P}_\#(\eps)
+\limsup_{n\to\infty}
   \Sfrac{1}{n}  \log p_n^\#(B_n) 
       \leq 2\log \C{P}_\#(\sfrac{1}{2}(\eps\plus \eps_M)) .
\end{equation}
Similarly to the above, taking $\eps\to\eps_M^-$ gives
$\limsup_{n\to\infty}\sfrac{1}{n} \log p_n^\#(B_n) 
\leq \log\C{P}_\#(\eps_M^-)$. If one \textit{defines} 
$\log \C{P}_\#(\eps_M) = \limsup_{n\to\infty}\sfrac{1}{n} \log p_n^\#(B_n)$
then this extends $\log\C{P}_\#(\eps)$ to a function which is concave on
to $[\eps_m,\eps_M]$.

This gives the following theorem.

\begin{theorem}
Define 
$\log \C{P}_\#(\eps_m) = \limsup_{n\to\infty}\sfrac{1}{n} \log p_n^\#(A_n)$
and
$\log \C{P}_\#(\eps_M) = \limsup_{n\to\infty}\sfrac{1}{n} \log p_n^\#(B_n)$.

Then the function $\log \C{P}_\#(\eps)$ is a concave function on 
$[\eps_m,\eps_M]$. \qed
\end{theorem}
\smallskip

Define $\gamma_n$ as the smallest value of $m$ so that 
$p_n(m) \leq p_n(\gamma_n)$ for all $A_n \leq m \leq B_n$.  Then
\begin{equation}
p_n(\lfl \eps n \rfl) 
\leq p_n \leq (B_n{-}A_n{+}1)\, p_n(\gamma_n) \leq  (B_n{-}A_n{+}1)\, p_n .
\end{equation}
Take logarithms and the divide by $n$, and the limit superior of
$\sfrac{1}{n} \log p_n(\gamma_n)$ as $n\to\infty$.  This proves that
there exists an accumulation point $\gamma\in [\eps_m,\eps_M]$ 
of $\{\gamma_n\}$, so that 
\begin{eqnarray}
& \hspace{-1.5cm}
\log \C{P}_\#(\eps)
\leq \limsup \sfrac{1}{n} \log p_n(\gamma_n) 
= \sup\{ \log \C{P}_\#(\eps) \svv \eps_m < \eps < \eps_M \} \nonumber \\
& = \log \C{P}_\#(\gamma) = \log \mu_\# . 
\label{22A}
\end{eqnarray}
By concavity of $\log \C{P}_\#(\eps)$ there are exists
$\eps_m \leq \eps_1 \leq \eps_2 \leq \eps_M$ so that
\begin{equation}
\log \C{P}_\#(\eps) 
\cases{
< \log \mu_\#  & \hbox{if $\eps_m \leq \eps < \eps_1$;} \\
= \log \mu_\#  & \hbox{if $\eps_1 \leq \eps \leq \eps_2$;} \\
< \log \mu_\#  & \hbox{if $\eps_2 < \eps \leq \eps_M$.} \\
} .
\label{24A}
\end{equation}
Moreover, $\log \C{P}_\#(\eps)$ is strictly increasing on $[\eps_m,\eps_2]$,
constant on $[\eps_1,\eps_2]$ and strictly decreasing on $[\eps_2,\eps_M]$.

\subsection{Integrated density functions}
\label{integrated density functions}

Define 
\begin{equation}
p_n({\leq}m) = \sum_{k=A_n}^m p_n(k),
\;\hbox{and}\;
p_n({\geq}m) = \sum_{k=m}^{B_n} p_n(k) .
\end{equation}
Define $\delta_{n,m}$ as that minimum value of $k\in\{A_n,\ldots,m\}$ so that
\begin{equation}
p_n(\delta_{n,m}) = \max_{A_n\leq k \leq m} p_n(k)
\end{equation}
and thus $p_n(m) \leq p_n(\delta_{n,m}) \leq p_n(\gamma_n)$. 
Since $\delta_{n,m} \leq m$ it follows, in particular, that 
\begin{equation}
p_n(m) \leq p_n(\delta_{n,m}) \leq p_n({\leq}m) \leq (m{-}A_n{+}1)\, p_n(\delta_{n,m}) .
\label{27A}
\end{equation}
Since $\log \C{P}_\#(\eps)$ is non-decreasing on $[\eps_m,\eps_2]$, it follows that, 
by putting $m=\lfl \eps n \rfl$,
\begin{equation}
\limsup_{n\to\infty} \sfrac{1}{n} \log p_n(\delta_{n,\lfl\eps n\rfl})
= \log \C{P}_\#(\eps), \; \hbox{for $\eps_m\leq \eps \leq \eps_2$}.
\end{equation}
Moreover, from equation \Ref{27A}, and the last equation, one may define
$\log \C{P}_\#({\leq}\eps)$ for $\eps \in [\eps_m,\eps_2]$,
\begin{equation}
\hspace{-2cm}
\log \C{P}_\#(\eps)
 = \lim_{n\to\infty} \sfrac{1}{n} \log p_n(\delta_{n,\lfl\eps n\rfl})
 = \lim_{n\to\infty} \sfrac{1}{n} \log p_n({\leq}\lfl \eps n \rfl)
 = \log \C{P}_\# ({\leq} \eps).
\label{29A}
\end{equation}
If $\eps \in (\eps_2,\eps_M]$, then
$p_n (\lfl \eps_2 m \rfl) \leq p_n({\leq} \lfl \eps n \rfl) \leq p_n(\gamma_n)$
and so by equations \Ref{22A} and \Ref{24A} it follows
\begin{equation}
\log \C{P}_\# ({\leq} \eps) = \log \mu_\#,
\;\hbox{for $\eps\in [\eps_2,\eps_M]$}.
\label{30A}
\end{equation}

The proof that $\log \C{P}_\#(\eps) = \lim_{n\to\infty} \sfrac{1}{n} 
\log p_n({\geq}\lfl \eps n \rfl) = \log \C{P}_\#({\geq} \eps)$ for
$\eps \in [\eps_1,\eps_M]$, and $\log \C{P}_\#({\geq} \eps) =
\lim_{n\to\infty} \sfrac{1}{n} \log p_n({\geq}\lfl \eps n \rfl) = 
\log \mu_\#$ if $\eps \in [\eps_m,\eps_1]$ is similar to the above.
Taking this together with equations \Ref{29A} and \Ref{30A},
and noting that for any $\eps \in [\eps_1,\eps_2]$
it is the case that $\log \C{P}_\# ({\leq}\eps) 
= \log \C{P}_\# ({\geq}\eps) = \log \mu_\#$, gives
the following theorem.

\begin{theorem}
Suppose that $\eps_1$ is the minimum value of $\eps$ such that
$\log \C{P}_\#(\eps) = \log \mu_\#$, and that $\eps_2$ is the 
maximum value of $\eps$ such that $\log \C{P}_\#(\eps) = \log \mu_\#$.
Then $\eps_m \leq \eps_1 \leq \eps_2 \leq \eps_M$.

For $\eps\in [\eps_m,\eps_M]$, the integrated density functions are given by
\[ \hspace{-1cm}
 \log \C{P}_\# ({\leq}\eps)
= \lim_{n\to\infty} \sfrac{1}{n} \log p_n({\leq}\lfl \eps n \rfl)
= \cases{ \log \C{P}_\#(\eps) ,& \hbox{if $\eps \in [\eps_m,\eps_1]$}; \cr
                \log \mu_\#, & \hbox{if $\eps \in [\eps_1,\eps_M]$},} \]
and
\[ \hspace{-1cm}
 \log \C{P}_\# ({\geq}\eps)
= \lim_{n\to\infty} \sfrac{1}{n} \log p_n({\leq}\lfl \eps n \rfl)
= \cases{ \log \mu_\#, & \hbox{if $\eps \in [\eps_m,\eps_2]$}; \cr
                \log \C{P}_\#(\eps) ,& \hbox{if $\eps \in [\eps_2,\eps_M]$} .} \]

Thus, it follows that for any $\eps \in [\eps_m,\eps_M]$,
\[\max\{\log \C{P}_\#({\leq}\eps), \log\C{P}_\#({\geq}\eps) \}
= \log \mu_\# \] 
and 
\[\min\{\log \C{P}_\#({\leq}\eps), \log\C{P}_\#({\geq}\eps) \}
= \log \C{P}_\#(\eps).\]
 \qed
\end{theorem}

\subsection{The free energy}
\label{free energy}

The partition function of a statistical mechanics model with $p_n^\#(m)$
states of size $n$ and energy $m$ is given by
\begin{equation}
Z_n(z) = \sum_{m=A_n}^{B_n} p_n^\#(m)\, z^m ,
\label{22}
\end{equation}
where $z = e^{-\xi/k_BT}>0$, $\xi$ is an interaction energy, $T$ the
absolute temperature and $k$ is Boltzmann's constant.  Since
$p_n^\#(m)$ satisfies equation \Ref{9}, it follows that
\begin{eqnarray*}
Z_{n_1}(z)\;Z_{n_2}(z) 
   &= \sum_{m_1=A_{n_1}}^{B_{n_1}} p_{n_1}^\#(m_1)\, z^{m_1}
          \sum_{m_2=A_{n_2}}^{B_{n_2}} p_{n_2}^\#(m_2)\, z^{m_2}
          \nonumber \\
   &\leq \LB B_{n_1+n_2} {-} A_{n_1+n_2} {+} 1 \RB  
    \sum_{m=A_{n_1+n_2}}^{B_{n_1+n_2}} p_{n_1+n_2}^\#(m)\, z^m 
          \nonumber \\
   &=  \LB B_{n_1+n_2} {-} A_{n_1+n_2} {+} 1 \RB\;   Z_{n_1+n_2} (z) .
\end{eqnarray*}
Thus, $\log Z_n(z)$ satisfies a superadditive relation of the type
in equation \Ref{6} and this proves existence of the limit
\begin{equation}
F(z) = \lim_{n\to\infty} \Sfrac{1}{n} \log Z_n(z) 
\label{23}
\end{equation}
by Hammersley's theorem \cite{H62}.  In this case it is said
that the model defined by equation \Ref{22} has a 
\textit{thermodynamic limit}, and $F(z)$ is the free energy.

It is the case that $F(z)$ is the Legendre transform 
of $\log \C{P}_\#(\eps)$.

\begin{theorem} 
If $\eps\in(\eps_m,\eps_M)$, then
$\displaystyle F(z) = \sup_{\eps\in(\eps_m,\eps_M)} 
\{\log \C{P}_\#(\eps) + \eps \log z\}$.
\label{theorem6}
\end{theorem}

\begin{proof}
For any $z>0$ there exists an $m$ in equation \Ref{22} maximizing 
the summand.  Denote the smallest such $m$ by $m_n$ (and $m_n$ 
is a function of $z$).  Then
\begin{equation}
p_n^\# (m_n)\, z^{m_n} \leq Z_n(z) 
\leq (B_n-A_n+1)\, p_n^\# (m_n)\, z^{m_n} .
\label{24}
\end{equation}
Define the function
\[ \eps_z = \limsup_{n\to\infty} \Sfrac{1}{n}\, m_n
= \lim_{k\to\infty} \Sfrac{1}{n_k} \, m_{n_k}, \]
where $\{ n_k\}$ is a subsequence realising the limit superior.  Put $n=n_k$
in equation \Ref{24}, and take $k\to\infty$. This shows that
\[ \lim_{k\to\infty} \Sfrac{1}{n_k} \log Z_{n_k} (z)
 = \lim_{k\to\infty} \Sfrac{1}{n_k} \log p_{n_k}^\# (m_{n_k})
   + \eps_z \log z . \]
By theorem \ref{theorem3.6},
\[ \lim_{k\to\infty} \Sfrac{1}{n_k} \log p_{n_k}^\# (m_{n_k})
 = \log \C{P}_\#(\eps_z) \]
 so that 
 \begin{equation}
\log \C{P}_\# (\eps_z) + \eps_z \log z
 =   \lim_{k\to\infty} \Sfrac{1}{n_k} \log Z_{n_k} (z)
 = F(z) 
 \label{25}
 \end{equation}
 by equation \Ref{23}.
 
By equation \Ref{22}, for all $z>0$ and $n\in\NatNu$, and for any 
$\eps \in (\eps_m,\eps_M)$,
\[ p_n^\#(\lfloor \eps n \rfloor)\, z^{\lfloor \eps n \rfloor} \leq Z_n(z). \]
Take logarithms, divide by $n$ and take $n\to\infty$.  This shows that
\begin{equation}
\log \C{P}_\#(\eps) + \eps \log z 
  \leq \lim_{n\to\infty} \Sfrac{1}{n} \log Z_n(z) = F(z) ,
  \quad\hbox{for any $\eps\in (\eps_m,\eps_M)$.}
\label{26}
\end{equation}
For fixed finite $z>0$, $\log \C{P}_\#(\eps)\plus \eps\log z$ is a 
concave and continuous function of $\eps\in(\eps_m,\eps_M)$.
Comparing equations \Ref{25} and \Ref{26} shows that
$\eps_z$ realises an equality in equation \Ref{26} so that
\[ F(z) = \sup_{\eps\in(\eps_m,\eps_M)} 
\{\log \C{P}_\#(\eps) + \eps \log z\}. \]
To see this, suppose that $\sup_{\eps\in(\eps_m,\eps_M)} 
\{\log \C{P}_\#(\eps) + \eps \log z\} < F(z)$ so that there
exists a $\delta>0$ such that 
$\log \C{P}_\#(\eps) + \eps \log z < F(z)\minus\delta$ 
for all  $\eps\in [\eps_m,\eps_M]$. This is a contradiction
of equation \Ref{25}. This completes the proof. 
\end{proof}
\smallskip

Next, by equation \Ref{26},
\[ \log \C{P}_\#(\eps) \leq \inf_{z>0} \{ F(z) - \eps \log z\} . \]
By the concavity of $\log \C{P}_\#(\eps)$, $\eps_z$ is a 
non-decreasing function of $z$.  Whenever it has an inverse 
$z_\eps$ it follows by equation \Ref{25} that
\begin{equation}
\log \C{P}_\#(\eps) = F(z_\eps) - \eps \log z_\eps .
\label{27}
\end{equation}
Similar to the argument above, if $\log \C{P}_\#(\eps) 
< \inf_{z>0} \{ F(z) \minus \eps \log z\}$ then there exists a 
$\delta>0$ such that $ \log \C{P}_\#(\eps) \plus \delta
< F(z) \minus \eps \log z$ for all $z\geq 0$. This is a contradiction
to equation \Ref{27}, and it follows that 
\begin{equation}
\log \C{P}_\#(\eps) = \inf_{z>0} \{ F(z) - \eps \log z\} . 
\label{28}
\end{equation}

\section{Some examples}
\label{examples}

\subsection{Adsorbing self-avoiding walks in the square lattice}

\subsubsection{Existence of the density function:}
The \textit{positive half square lattice} is the subset of the square lattice
with non-negative $y$-coordinates.  The \textit{boundary} or 
\textit{adsorbing line} in the positive half square lattice is the $x$-axis.
A \textit{loop} in the positive half square lattice is a self-avoiding walk
with both its endpoints in the adsorbing line.  See the left
panel of figure \ref{figure2}.  An \textit{adsorbing loop} is a loop with 
a given number $v$ of vertices in the adsorbing line.  The number of 
adsorbing loops of length $n$ with $v$ visits to the adsorbing line is 
denoted by $\ell_n(v)$.

A loop is \textit{unfolded} if its end-vertices are left-most and
right-most; see the right panel of figure \ref{figure2} for an example
of an unfolded loop. The number of unfolded loops of length $n$ is 
denoted by $\ell_n^\dagger$. By the Hammersley-Welsh construction \cite{HW62A},
\begin{equation}
\ell_n^\dagger (v) \leq \ell_n(v) \leq e^{o(n)}\, \ell_n^\dagger (v).
\label{29}
\end{equation}
By concatenating an unfolded loop of length $n$ with $v$ visits, and 
an unfolded loop of length $m$ with $w$ visits, as shown in 
figure \ref{figure3}, $\ell_n^\dagger (v) \; \ell_m^\dagger (w) 
\leq \ell_{n+m+1}^\dagger (v+w)$
since the resulting adsorbing loop has length $n{+}m{+}1$.
If $n$ is replaced by $n{-}1$, and $m$ by $m{-}1$, then
\begin{equation}
\ell_{n-1}^\dagger (v) \; \ell_{m-1}^\dagger (w) 
\leq \ell_{n+m-1}^\dagger (v+w) 
\end{equation}
and $\ell_{n-1}^\dagger(v)$ satisfies equation \Ref{10}.
Since $\ell^\dagger_{n-1}(v) \leq \ell_{n-1}(v) \leq 4^{n-1}$ and, 
for a given $n{-}1$, $0\leq v \leq n$, if follows by theorems \ref{theorem5}
and \ref{theorem3.6}, and by equation \Ref{29}, that
\begin{equation}
\hspace{-2cm}
\log \C{P}_v (\epsilon) 
= \lim_{n\to\infty} \sfrac{1}{n} \log \ell_n^\dagger (\lfl \epsilon n \rfl)
= \lim_{n\to\infty} \sfrac{1}{n} \log \ell_n (\lfl \epsilon n \rfl),
\;\hbox{for $\eps\in(0,1)$}.
\label{31}
\end{equation}

\begin{figure}[t!]
\beginpicture
\setcoordinatesystem units <1pt,1pt> 
\setlinear

\setplotarea x from -100 to 250, y from -20 to 35

\setcoordinatesystem units <1pt,1pt> point at -70 20

\put
{\beginpicture
\setplotsymbol ({$\cdot$})
\plot  -30 0  -40 0  -50 0  -50 10  -50 20  -40 20  -40 30  -30 30  -30 20 
-30 10  -20 10  -20 20  -10 20  -10 10  -10 0  0 0  10 0  10 10  10 20  
20 20  20 10  20 0   /
\plot 60 0  60 10   60 20  60 30  70 30  80 30  90 30  90 20  80 20  
70 20  70 10  70 0  80 0  90 0  90 10  100 10  100 20  110 20  110 10 
110 0  120 0  120 10  130 10  130 0   /  
\multiput {\fns$\bullet$} at
  -30 0  -40 0  -50 0  -50 10  -50 20  -40 20  -40 30  -30 30  -30 20 
-30 10  -20 10  -20 20  -10 20  -10 10  -10 0  0 0  10 0  10 10  10 20  
20 20  20 10  20 0  60 0  60 10   60 20  60 30  70 30  80 30  90 30  90 20  80 20  
70 20  70 10  70 0  80 0  90 0  90 10  100 10  100 20  110 20  110 10 
110 0  120 0  120 10  130 10  130 0  /
 
\put {$\ell_n (v)$} at -10 35 
\put {$\ell_n^\dagger (v)$} at 110 35 

\setplotsymbol ({\scalebox{2.0}{.}})

\plot -60 -4 30 -4 /
\plot 50 -4 140 -4 /

\endpicture} at 0 30 

\endpicture
\caption{\textit{An adsorbing loop (on the left), and an unfolded adsorbing
loop on the right).}}
\label{figure2}
\end{figure}

Let $c_n$ be the number of self-avoiding walks from the origin in the square lattice,
and let $c_n^+(v)$ be the number of self-avoiding walks, from the origin 
in the positive half square lattice (\textit{positive walks}), of length $n$ with 
$v$ visits to the adsorbing line.   If the height ($y$-coordinate) of 
the last vertex of a positive walk is denoted by $h$, then 
$c_n^+(v;h)$ is the number of positive walks of length $n$, 
with $v$ visits and last vertex at height $h$.

\begin{figure}[t!]
\beginpicture
\setcoordinatesystem units <1pt,1pt> 
\setlinear

\setplotarea x from -100 to 250, y from -20 to 30

\setcoordinatesystem units <1pt,1pt> point at -70 20

\put
{\beginpicture
\setplotsymbol ({$\cdot$})
\plot -50 0  -40 0  -40 10  -30 10  -30 20  -20 20  -20 10  -20 0  -10 0  -10 10  
0 10  0 20  10 20  10 10  10 0 /
\plot 20 0  20 10  20 20  20 30  30 30  30 20  40 20  40 10  30 10  30 0
40 0  50 0  50 10  60 10  70 10  70 0  /
\multiput {\fns$\bullet$} at
-50 0  -40 0  -40 10  -30 10  -30 20  -20 20  -20 10  -20 0  -10 0  -10 10  
0 10  0 20  10 20  10 10  10 0  20 0  20 10  20 20  20 30  30 30  30 20  
40 20  40 10  30 10  30 0 40 0  50 0  50 10  60 10  70 10  70 0 /

\setplotsymbol ({.})

\plot 14 0 16 0 /
 
\put {$\ell_n^\dagger (v)$} at -20 30 
\put {$\ell_m^\dagger (w)$} at 60 25 

\setplotsymbol ({\scalebox{2.0}{.}})

\plot -60 -4 80 -4 /

\endpicture} at 0 30 

\endpicture
\caption{\textit{Concatenating two unfolded adsorbing self-avoiding loops.
If the loop on the left has length $n$ and $v$ visits to the adsorbing plane,
then the number of conformations is $\ell_n^\dagger (v)$.  Similarly, the
loop on the right has length $m$ and $w$ visits, so the number of conformations
is $\ell_m^\dagger (w)$. Concatenating these loops by adding a step between
them shows that $\ell_n^\dagger(v)\,\ell_m^\dagger(w) 
\leq \ell_{n+m+1}^\dagger(v+w)$.}}
\label{figure3}
\end{figure}

\textit{Unfolded positive walks} are positive walks, with first vertex 
(at the origin) left-most, and last vertex right-most.  Denote by 
$c_n^{{+,\dagger}}(v;h)$ the number of unfolded positive walks of
length $n$, and with $v$ visits and last vertex at height $h$. 

Positive walks can be unfolded, using the Hammersley-Welsh construction 
\cite{HW62A}, while keeping constant both $v$ and $h$.  This shows that
\begin{equation}
c_n^{+,\dagger} (v;h) \leq 
c_n^+(v;h) \leq e^{o(n)}\,c_n^{+,\dagger} (v;h) .
\label{32}
\end{equation}
For each of the functions $c_n^{+,\dagger} (v;h)$ and $c_n^+(v;h)$ there
are, for given $n$ and $v$, most popular values of $h$, denoted
by $h^*$ (a function of $n$ and $v$), so that
\begin{equation}
c_n^{+,\dagger} (v;h^*) \leq c_n^{+,\dagger}(v)
\leq (n{+}1)\, c_n^{+,\dagger}(v,h^*)
\label{33}
\end{equation}
since there are at most $n{+}1$ different choices for $h$.
By unfolding positive walks with $v$ visits using the Hammersley-Welsh
construction \cite{HW62A} gives 
$c_n^+(v) \leq e^{o(n)}\, c_n^{+,\dagger}(v)$, and by combining 
this with equations \Ref{32} and \Ref{33},
\begin{equation}
\hspace{-2cm}
c_n^{+,\dagger}(v;h) \leq c_n^{+,\dagger}(v;h^*)
\leq c_n^+(v) 
\leq e^{o(n)}\, c_n^{+,\dagger}(v) 
\leq (n{+}1)\,e^{o(n)}\,c_n^{+,\dagger}(v;h^*) .
\label{34}
\end{equation}

Two positive unfolded walks, with endpoints of height $h$, can be
concatenated as shown in figure \ref{figure4} (by reflecting
the second walk in the vertical plane and placing it as shown)
to create an adsorbing loop.  Choosing $h$ to be the most
popular height of unfolded positive adsorbing walks give
\begin{equation}
\LB c_n^{+,\dagger}(v;h^*) \RB^2 
\leq \ell_{2n+1}(2v)
\leq c_{2n+1}^+(2v) .
\label{35}
\end{equation}
Combining equations \Ref{34} and \Ref{35} and noting the
existence of the limits in equation \Ref{31} proves that the
density function of positive adsorbing self-avoiding walks is given by
\begin{equation}
\hspace{-2cm}
\log \C{P}_v (\epsilon) 
= \lim_{n\to\infty} \sfrac{1}{n} \log c_n^+ (\lfl \epsilon n \rfl)
= \lim_{n\to\infty} \sfrac{1}{n} \log \ell_n^\dagger (\lfl \epsilon n \rfl),
\;\hbox{for $\eps\in(0,1)$}.
\end{equation}

\subsubsection{Numerical estimation of $\log \C{P}_v(\epsilon)$:}
Positive walks were approximately enumerated using the GAS-algorithm
\cite{JvRR09,JvR16} implemented with Berretti-Sokal elementary moves 
\cite{BS85}.  A total of $5\times 10^{11}$ iterations were performed
sampling positive walks of lengths $n\in\{0,1,2,\ldots,500\}$, and the data
were used to estimate $c_n^+(v)$.  Since $\log \C{P}_v(\eps)$ is approximated
by $\sfrac{1}{n} \log c_n (v/n)$ when $v = \lfl \eps n \rfl$ and $n$ is large,
a plot of $\sfrac{1}{n} \log c_n (v/n)$ against $v/n$ is shown in figure
\ref{figure5}.  The data points are accumulating on a limiting curve, which
as $n\to\infty$, will be $\log \C{P}_v(\eps)$.

\begin{figure}[t!]
\beginpicture
\setcoordinatesystem units <1pt,1pt> 
\setlinear

\setplotarea x from -100 to 250, y from -20 to 35

\setcoordinatesystem units <1pt,1pt> point at -70 20

\put
{\beginpicture

\arrow <6pt> [.2,.67] from 15 0 to 15 30 

\setplotsymbol ({$\cdot$})
\plot -50 0  -40 0  -40 10  -30 10  -30 20  -20 20  -20 10  -20 0  -10 0  -10 10  
0 10  0 20  10 20  10 30 /
\plot  20 30  30 30  30 20  40 20  40 10  30 10  30 0
40 0  50 0  50 10  60 10  70 10  80 10  80 0 /
\multiput {\fns$\bullet$} at
-50 0  -40 0  -40 10  -30 10  -30 20  -20 20  -20 10  -20 0  -10 0  -10 10  
0 10  0 20  10 20  10 30  20 30  30 30  30 20  40 20  40 10  30 10  30 0
40 0  50 0  50 10  60 10  70 10  80 10  80 0  /

\setplotsymbol ({.})

\plot 14 30 16 30 /
 
\put {$c_n^{+,\dagger} (v)$} at -20 30 
\put {$h$} at 20 15
\put {$c_n^{+,\dagger} (v)$} at 60 25 

\setplotsymbol ({\scalebox{2.0}{.}})

\plot -60 -4 90 -4 /

\endpicture} at 0 30 

\endpicture
\caption{\textit{Concatenating two unfolded positive walks to 
create an adsorbing loop.}}
\label{figure4}
\end{figure}

The partition function of adsorbing positive walks is given by
\begin{equation}
V_n(a) = \sum_{v=0}^n c^+_n(v)\, a^v
\end{equation}
where $a=e^{-1/k_BT}$ is an activity conjugate to the number of visits.
The limiting free energy
\begin{equation}
\kappa(a) = \lim_{n\to\infty} \Sfrac{1}{n} \log V_n(a)
\end{equation}
exists, and $\kappa(a) \simeq \log \mu_{d-1} + \log a$ \cite{HTW82}
(where $\mu_d$ is the growth constant of the self-avoiding walk).
It is reasonable to assume that $\kappa(a) \simeq \log\mu_{d-1} + \log a  + A/\log a$
for large $a>1$.   By equation \Ref{28} one may then estimate
$\log \C{P}_v(\eps)$ by minimizing $\log \mu_{d-1} + \log a + A/\log a - \epsilon \log a$.
This gives  
\begin{equation}
\log \C{P}_v(\eps) \approx \log \mu_{d-1} + 2\sqrt{A}\, \sqrt{1-\epsilon} .
\end{equation}
For the purpose of fitting this to numerical data, it would be reasonable to 
assume that 
\begin{equation}
\hspace{-1cm}
\log \C{P}_v(\eps) \simeq \alpha + \beta \sqrt{1-\eps} + \Sfrac{\gamma}{\sqrt{1-\eps}}
+ \Sfrac{\delta}{1-\eps} +  \cdots, 
\;\hbox{for $\eps\in(0,1)$} .
\label{40}   
\end{equation}
for $\eps\in(0,1)$ with higher order decaying terms left away.  Since 
$\log \C{P}_v(0^+)  = \log \mu_d$, it should also be the case
that $\alpha + \beta + \gamma + \delta + \cdots = \log \mu_d$. 
Also, $\log \C{P}_v(1) = 0$, since the fully adsorbed phase in the half
square lattice has growth constant $\log \mu_1 = 0$.  

A best fit to our data gives
\begin{equation}
\hspace{-2cm}
\log \C{P}_v(\eps) \approx 
-0.256
+ 1.202\, \sqrt{1-\eps} 
+ \frac{0.0300}{\sqrt{1-\eps}}
+ \frac{0.00135}{1-\eps}  .
\end{equation}
Putting $\eps=0$ in the above gives $\C{P}(0) \approx 0.97069$, 
which compares well to the accurate estimate 
$\log \mu_2 \approx 0.97008114726(76)$ \cite{CJ12}.  By taking the 
right derivative of $\log \C{P}(\eps)$ to $\eps$ at $\eps=0$, the estimated 
adsorption critical point is
\begin{equation}
\log a_c^+ = 0.587\ldots, \qd\hbox{in the square lattice}.
\end{equation}
This compares well with the accurate estimate $\log a^+_c = 0.57416$ 
in reference \cite{BGJ12} and $\log a^+_c = 0.5761(17)$ in 
reference \cite{JvR16}.

\begin{figure}[t]
\centering
\includegraphics[scale=0.4]{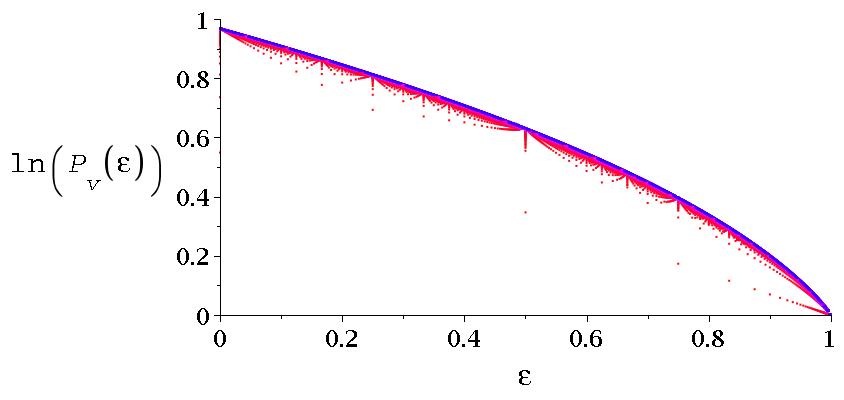}
\caption{$\log \C{P}_v (\eps)$ as a function of $\eps\in(0,1)$.
The points have coordinates $(v/n,(\log c_n^+(v))/n)$ for
$0\leq n \leq 500$ and $0\leq v \leq n$ and they accumulate on
$\log \C{P}_v(\eps)$ as $n$ increases, by theorem \ref{theorem3.6}.}
\label{figure5}   
\end{figure}

In the cubic lattice a similar least squares analysis of data for 
adsorbing walks of length up to $n=500$ gives
\begin{equation}
\hspace{-2cm}
\log \C{P}_v(\eps) \approx 
0.922
+ 0.628\, \sqrt{1-\eps} 
+ \frac{0.00469}{\sqrt{1-\eps}}
+ \frac{0.000639}{1-\eps}  .
\end{equation}  This gives the estimate $\C{P}(0) \approx 1.54589$, 
which is close to $\log \mu_3 \approx 1.544160971(58)$ \cite{C13}.  
The estimate for the critical point is 
\begin{equation}
\log a^+_c = 0.316\ldots, \qd\hbox{in the cubic lattice}.
\end{equation}
This estimate is slightly larger than the estimate 
$\log a^+_c = 0.2670(54)$ in reference \cite{JvR16}, which due to the
methods in that paper, is the more accurate estimate.

\subsection{Adsorbing polygons in the square lattice}

\subsubsection{Existence of the density function:}
Adsorbing polygons is a model of adsorbing ring polymers.   Denote
the number of square lattice polygons of length $n$, counted up to
translations, by $p_n$.  Notice that the square lattice is bipartite,
so that $p_n=0$ if $n$ is odd.   We assume that $n$ and $m$ are both
even in what follows below.

The lattice for this model is again the (positive) half square lattice, but 
with lattice polygons which are constrained to have at least one vertex
(a \textit{visit}) in the adsorbing line.  Denote the number of these
\textit{adsorbing polygons} of length $n$, (counted under 
equivalence of translations parallel to the adsorbing line), with 
$v$ vertices in the adsorbing line, by $p_n(v)$.  The partition
function is defined by
\begin{equation}
P_n(a) = \sum_{v=2}^{n/2} p_n(v)\, a^v,
\end{equation}
since each adsorbing polygon has at least two visits in the adsorbing
line, and cannot have more than $n/2$ visits in the half square lattice.

An indirect proof of the existence of the free energy of adsorbing square
lattice polygons was given in reference \cite{JvR15} (see section 9.2), using
most popular class arguments \cite{HW62A} to prove that
a large subclass of adsorbing polygons satisfies a supermultiplicative
relation.

\begin{figure}[t!]
\beginpicture
\setcoordinatesystem units <1pt,1pt> 
\setlinear

\setplotarea x from -120 to 250, y from -20 to 55

\setcoordinatesystem units <1pt,1pt> point at -70 20

\put
{\beginpicture

\arrow <5pt> [.2,.67] from -80 0 to -80 10 
\arrow <5pt> [.2,.67] from -80 0 to -70 0 
\put {$\hj$} at -85 12
\put {$\hi$} at -72 -10

\setplotsymbol ({$\cdot$})
\plot  -30 0  -40 0  -50 0  -50 10  -50 20  -40 20  -40 30  -30 30  -30 20 
-30 10  -20 10  -20 20  -10 20  -10 30  -10 40   0 40 0 30  10 30  10 20
10 10  20 10  30 10  40 10  40 0  30 0  20 0  10 0  0 0  -10 0  -20 0  -30 0    /
\plot  50 0  50 10  50 20  60 20  60 30  80 30  80 20  80 10  
90 10  90 20  100 20  100 30  110 30 120 30  120 20
130 20  130 10  130 0  120 0  110 0  110 10  100 10  100 0
90 0  80 0  70 0  70 10  60 10  60 0  50 0   /
\multiput {\fns$\bullet$} at
 -30 0  -40 0  -50 0  -50 10  -50 20  -40 20  -40 30  -30 30  -30 20 
-30 10  -20 10  -20 20  -10 20  -10 30  -10 40   0 40 0 30  10 30  10 20
10 10  20 10  30 10  40 10  40 0  30 0  20 0  10 0  0 0  -10 0  -20 0  
-30 0   50 0  50 10  50 20  60 20  60 30  70 30  80 30  80 20  80 10  
90 10  90 20  100 20  100 30  110 30  120 30  
120 20  130 20  130 10  130 0  120 0  110 0  110 10  100 10  
100 0  90 0  80 0  70 0  70 10  60 10  60 0  50 0 /

\plot 110 34 110 40  120 40 120 34    /  
 
\put {$p_n (v)$} at -30 40 
\put {$p_m (w)$} at 80 40 
\put {$A$} at 130 40

\plot  44 0 46 0 /
\plot  44 10 46 10 /

\plot 37 3 43 7 /
\plot 47 3 53 7 /

\setplotsymbol ({\scalebox{2.0}{.}})

\plot -60 -4 140 -4 /

\put {(a)} at -110 20

\endpicture} at 0 30

\setcoordinatesystem units <1pt,1pt> point at -70 90

\put
{\beginpicture

\arrow <5pt> [.2,.67] from -80 0 to -80 10 
\arrow <5pt> [.2,.67] from -80 0 to -70 0 
\put {$\hj$} at -85 12
\put {$\hi$} at -72 -10

\setplotsymbol ({$\cdot$})
\plot  -30 0  -40 0  -50 0  -50 10  -50 20  -40 20  -40 30  -30 30  -30 20 
-30 10  -20 10  -20 20  -10 20  -10 30  -10 40   0 40 0 30  10 30  20 30  
20 20  30 20  30 10  40 10  40 0   30 0  20 0  20 10  10 10  0 10
0 0  -10 0  -20 0  -30 0    /
\plot  60 40 70 40  80 40  90 40  90 30  80 30  80 20  80 10  
90 10  90 20  100 20  100 30  110 30  110 40  120 40  120 30  
120 20  130 20  130 10  130 0  120 0  110 0  110 10  100 10  100 0
90 0  80 0  70 0  60 0  60 10  50 10  50 20  40 20  40 30  30 30 
30 40  40 40  50 40  60 40 /
\multiput {\fns$\bullet$} at
-30 0  -40 0  -50 0  -50 10  -50 20  -40 20  -40 30  -30 30  -30 20 
-30 10  -20 10  -20 20  -10 20  -10 30  -10 40   0 40 0 30  10 30  20 30  
20 20  30 20  30 10  40 10  40 0   30 0  20 0  20 10  10 10  0 10
0 0  -10 0  -20 0  -30 0  60 40 70 40  80 40  90 40  90 30  80 30  80 20  
80 10  90 10  90 20  100 20  100 30  110 30  110 40  120 40  120 30  
120 20  130 20  130 10  130 0  120 0  110 0  110 10  100 10  100 0
90 0  80 0  70 0  60 0  60 10  50 10  50 20  40 20  40 30  30 30 
30 40  40 40  50 40  60 40  /
 
\put {$p_n (v)$} at -30 40 
\put {$p_m (w)$} at 80 50 
\put {$A$} at 130 50

\plot  24 30 26 30 /

\plot  110 44  110 50  120 50  120 44  /

\setplotsymbol ({\scalebox{2.0}{.}})

\plot -60 -4 140 -4 /

\put {(b)} at -110 20

\endpicture} at 0 30 

\setcoordinatesystem units <1pt,1pt> point at -70 180

\put
{\beginpicture

\arrow <5pt> [.2,.67] from -80 0 to -80 20 
\arrow <5pt> [.2,.67] from -80 0 to -60 0 
\put {$\hj$} at -86 12
\put {$\hi$} at -70 -7

\setplotsymbol ({$\cdot$})
\plot  -20 0 0 0 20 0 20 20 20 40 40 40 40 60 20 60 0 60 
-20 60 -20 40 -20 20 -20 0  /
\plot  100 20 100 0 80 0 60 0 60 20 60 40 80 40 80 60 
100 60 100 40 100 20   /
\multiput {\fns$\bullet$} at
-20 0 0 0 20 0 20 20 20 40 40 40 40 60 20 60 0 60 
-20 60 -20 40 -20 20 -20 0
100 20 100 0 80 0 60 0 60 20 60 40 80 40 80 60 
100 60 100 40 100 20  /
 
\plot  45 40 55 40 /

\multiput {\scalebox{2.0}{$*$}} at 40 40 60 40  /

\put {$\h{x}$} at 35 48
\put {$\h{y}$} at 67 47

\multiput {\scalebox{1.5}{$\circ$}} at 60 60 40 20 /

\setplotsymbol ({\scalebox{2.0}{.}})

\plot -50 -4 130 -4 /

\put {(c)} at -110 20

\endpicture} at 0 30 
\normalcolor

\endpicture
\caption{\textit{(a) Two adsorbing square lattice polygons placed to that there
is a pair of parallel edges between them.  (b)  These two polygons cannot be placed
so that there is a pair of parallel edges between them.  (c) The highest
pair of nearest neighbour vertices $(\h{x},\h{y})$ in two polygons.
Notice that either $\h{x}+\h{j}$, or $\h{y}+\h{j}$, but not both, are accupied.}}
\label{figure6}
\end{figure}

In this section a direct proof is given showing that adsorbing polygons satisfy 
a supermultiplicative inequality.   This proves existence of $\log P_p(\eps)$
(which is a new result).  This proof proceeds by first showing that 
two adsorbing polygons can be translated in the horizontal direction 
so that they are close to each other (see figure \ref{figure6}(a)) and 
then they can be  concatenated into a single polygon.  

There are several cases to consider in completing this construction.

In case (a) it is possible to translate a polygon of length $n$
with $v$ visits and a polygon of length $m$ with $w$ visits, so that
there is a pair of parallel vertical edges between them. By deleting these
two edges, and inserting two edges, as shown in figure \ref{figure6}(a),
one obtains a polygon of length $n{+}m$ with $v{+}w$ visits.
Notice that the concatenated polygon has a top-most right-most edge
which can be replaced by three edges as shown at $A$ in figure
\ref{figure6}(a).  This increases the length of the polygon by two
without creating new visits to the adsorbing line.

There are cases where it is not possible to translate two polygons
so that they can be concatenated as in case (a) above.  An example
is shown in figure \ref{figure6}(b).  It is however, possible to translate
the two polygons so that there is a pair of vertices, one in the first, and the
other in the second polygon, which are one step apart in the horizontal
direction.  There is, moreover, a highest pair of such vertices (shown
in figure \ref{figure6}(b)).  Let these vertices be denoted $\h{x}$ and
$\h{y}$, shown in figure \ref{figure6}(c).

Having identified the highest pair of vertices $(\h{x},\h{y})$, where
$\h{x}$ is in the first polygon, and $\h{y}$ is in the second polygon,
so that $\h{y}=\h{x}{+}\hi$, it is now shown how to concatenate 
the two polygons. 

First examine the vicinity of the pair of vertices $\h{x}$ and $\h{y}$.  
Observe that one may not have both vertices $\h{x}{+}\hj$ and $\h{y}{+}\hj$
occupied (since $\h{x}$ and $\h{y}$ are the highest such vertices).
In addition, if both $\h{x}{+}\hj$ and $\h{y}{+}\hj$ are unoccupied,
then both the edges $\edge{\h{x}}{(\h{x}{-}\hj)}$ and
$\edge{\h{x}}{(\h{x}{-}\hj)}$ are in the polygons respectively,
and case (a) is recovered.

Thus, exactly one of $\h{x}{+}\hj$ and $\h{y}{+}\hj$ are occupied
by the respective polygons, and we assume, without loss of generality 
that this is $\h{x}{+}\hj$.  The other case can be dealt with by
reflecting the polygons in the vertical line, completing the concatenation,
and then reflecting it back.

By analysing all possible arrangements of edges incident on
$\h{x}$ and $\h{y}$, it is seen that the edge $\edge{\h{x}}{(\h{x}{+}\hj)}$ 
must be in the first polygon, while the edge $\edge{\h{y}}{(\h{y}{+}\hj)}$ 
cannot be in the first or second polygon, since $\h{y}{+}\hj$ is not occupied.

This gives the situation in figure \ref{figure7}(a)(1) and (a)(2).
In other words, the basic starting point is the arrangement of edges
about the pair of vertices $(\h{x},\h{y})$ shown in figure \ref{figure7}(a)(2).

Notice that the height of the vertices $(\h{x},\h{y})$ is at least one 
(step above the adsorbing line) -- they cannot be visits.  We will consider
two cases.  In the first instance the edge $\edge{\h{x}}{\h{y}}$ is at 
height one (thus close to the adsorbing line), and secondly, the height 
of $\edge{\h{x}}{\h{y}}$ is at least two.

\begin{figure}[t!]
\beginpicture
\setcoordinatesystem units <1pt,1pt> 
\setlinear

\setplotarea x from -120 to 250, y from -20 to 55

\setcoordinatesystem units <1pt,1pt> point at -35 20

\put
{\beginpicture

\arrow <5pt> [.2,.67] from -80 0 to -80 20 
\arrow <5pt> [.2,.67] from -80 0 to -60 0 
\put {$\hj$} at -87 12
\put {$\hi$} at -72 -9

\setplotsymbol ({$\cdot$})

\setdots <4pt>
\plot -40 20  -20 20 /
\setsolid
\plot -40 20  -40 40 /

\multiput {$\bullet$} at -40 20  -20 20  -40 40  /
\put {$\circ$} at -20 40 

\put {$\h{x}$} at -46 14
\put {$\h{y}$} at -14 13

\put {(1)} at -30 -14

\setplotsymbol ({\fiverm.})
\arrow <5pt> [.2,.67] from 0 20  to  20 20 

\setplotsymbol ({$\cdot$})

\setdots <4pt>
\plot 60 20  80 20 /
\setsolid
\plot 60 20  60 40 /
\plot 80 0  80 20  100 20 /
\plot 40 20  60 20 /

\multiput {$\bullet$} at 60 20  80 20  60 40  80 0  100 20  40 20 /
\put {$\circ$} at 80 40 

\put {$\h{x}$} at 54 12
\put {$\h{y}$} at 86 13

\put {$\times$} at 60 0
\put {${+}$} at 100 40

\put {(2)} at 70 -14

\setplotsymbol ({\scalebox{2.0}{.}})
\color{white}
\multiput {$\bullet$} at -60 0 140 0 /
\color{black}
\put {(a)} at -110 20

\endpicture} at 0 30

\setcoordinatesystem units <1pt,1pt> point at -69 110

\put
{\beginpicture

\arrow <5pt> [.2,.67] from -90 0 to -90 20 
\arrow <5pt> [.2,.67] from -90 0 to -70 0 
\put {$\hj$} at -97 12
\put {$\hi$} at -82 -9

\setplotsymbol ({$\cdot$})

\setdots <4pt>
\setsolid
\plot -60 20  -40 20  -40 40 /
\plot 0 0 -20 0  -20 20  0 20 /
\setdashes <2pt>
\plot -40 20  -20 20  /
\plot -40 40 -20 40 0 40 0 20 /
\setsolid
\multiput {$\bullet$} at -60 20  -40 20  0 0  -20 0  -20 20  0 20  -40 40  /
\multiput {$\circ$} at -20 40  -40 0 / 

\put {$\h{x}$} at -46 14
\put {$\h{y}$} at -14 13

\plot -44 34  -36 26  /
\plot -14 24  -6 16 /

\put {$\circ$} at 0 40

\put {(3)} at -30 -16

\setplotsymbol ({\fiverm.})

\setplotsymbol ({$\cdot$})

\setdots <4pt>
\setsolid
\plot 20 20  40 20  40 40 /
\plot 80 0 60 0  60 20  80 20 80 40 /
\setdashes <2pt>
\plot 40 20  60 20  /
\plot 40 40 60 40 80 40 /
\setsolid
\multiput {$\bullet$} at 20 20  40 20  80 0  60 0  60 20  80 20  40 40 80 40 /
\multiput {$\circ$} at 60 40  40 0 / 

\put {$\h{x}$} at 34 14
\put {$\h{y}$} at 66 13

\plot 36 34  44 26  /
\plot 76 34  84 26 /
\plot 66 24  74 16 /

\put {(4)} at 50 -16

\setplotsymbol ({\fiverm.})

\setplotsymbol ({$\cdot$})

\setdots <4pt>
\plot 120 20  140 20 /
\setsolid
\plot 100 20  120 20  120 40 /
\plot 200 0 180 0 160 0 140 0  140 20  160 20 180 20   /
\plot 160 60 160 40 180 40 /
\multiput {$\bullet$} at 100 20  120 20  120 40 
                                       160 0 140 0  140 20  160 20 
                                       160 60 160 40 180 40  180 20  180 0 200 0 200 20 /
\multiput {$\circ$} at 140 40  120 0 /

\put {$\alpha$} at 180 30 
\put {$\beta$} at 190 20 

\plot 180 20 180 25 /  \plot 180 40 180 35 /
\plot 180 20 185 20 /  \plot 200 20 195 20 /

\put {$\h{x}$} at 114 14
\put {$\h{y}$} at 146 13

\put {(5)} at 130 -16

\setplotsymbol ({\scalebox{2.0}{.}})

\plot -60 -4 0 -4 /
\plot 20 -4 80 -4 /
\plot 100 -4 200 -4 /

\put {(b)} at -120 20

\endpicture} at 0 30

\setcoordinatesystem units <1pt,1pt> point at -60 210

\put
{\beginpicture

\arrow <5pt> [.2,.67] from -90 0 to -90 20 
\arrow <5pt> [.2,.67] from -90 0 to -70 0 
\put {$\hj$} at -97 12
\put {$\hi$} at -82 -9

\setplotsymbol ({\fiverm.})

\setplotsymbol ({$\cdot$})

\setdots <4pt>
\plot -40 20  -20 20 /
\setsolid
\plot -60 20  -40 20  -40 40 /
\plot 40 0 20 0 0 0 -20 0  -20 20  0 20 20 20 40 20  /
\plot 0 60 0 40 20 40 /
\setdashes <2pt>
\plot 40 20  40 0  /
\setsolid
\multiput {$\bullet$} at  -60 20  -40 20  -40 40
                                        40 0 20 0 0 0 -20 0  -20 20  0 20 20 20 40 20
                                        0 60 0 40 20 40  /
\multiput {$\circ$} at -20 40  -40 0 / 

\put {$\h{x}$} at -46 14
\put {$\h{y}$} at -14 13

\plot -24 14 -16 6 /
\plot -14 4 -6 -4 /
\plot 6 4 14 -4 /
\plot 26 4 34 -4 /
\plot -14 24 -6 16 /
\plot 6 24 14 16 /
\plot 26 24 34 16 /

\put {(6)} at -10 -16

\setplotsymbol ({\fiverm.})
\arrow <5pt> [.2,.67] from 50 20  to  70 20 

\setplotsymbol ({$\cdot$})

\setdots <4pt>
\setsolid
\plot 80 20  100 20  100 40 /
\plot 180 0 180 20  /
\plot 140 60 140 40 160 40 /
\setdashes <2pt>
\plot 100 40  120 40 140 40  /
\plot 100 20 100 0 120 0 140 0 140 20 160 20 160 40  /
\setsolid
\multiput {$\bullet$} at  80 20  100 20  100 40
                                       180 0 180 20
                                       140 60 140 40 160 40   /
\multiput {$\circ$} at 100 0  120 40 120 0 120 20 
                                   140 0  140 20  160 0  160 20
                                   180 0  180 20  / 

\put {$\h{x}$} at 94 14
\put {$\h{y}$} at 126 13

\plot 96 34 104 26 /
\plot 146 44 154 36 /

\put {(7)} at 130 -16

\setplotsymbol ({\scalebox{2.0}{.}})

\plot -60 -4 40 -4 /
\plot 80 -4 180 -4 /

\put {(c)} at -120 20

\endpicture} at 0 30

\setcoordinatesystem units <1pt,1pt> point at -37 310

\put
{\beginpicture

\arrow <5pt> [.2,.67] from -80 0 to -80 20 
\arrow <5pt> [.2,.67] from -80 0 to -60 0 
\put {$\hj$} at -87 12
\put {$\hi$} at -72 -9

\setplotsymbol ({$\cdot$})

\setdots <4pt>
\plot -40 20 -20 20 /
\setsolid
\plot -60 20 -40 20  -40 40 /
\plot -20 0  -20 20  0 20 /
\setdashes <2pt>
\plot -40 20  -40 0  -20 0 /
\plot -40 40  -20 40  -20 20 /
\setsolid
\plot -44 32 -36 28 /
\plot -24 12 -16 8 /

\multiput {$\bullet$} at -60 20  -40 20  -20 0  -20 20  0 20  -40 40  /
\multiput {$\circ$} at -20 40  -40 0 / 

\put {$\h{x}$} at -46 14
\put {$\h{y}$} at -12 14

\put {(8)} at -30 -34

\setplotsymbol ({\fiverm.})

\setplotsymbol ({$\cdot$})

\setdots <4pt>
\plot 40 20  60 20 /
\setsolid
\plot 40 20  40 40 /
\plot 40 0  60 0  60 20  80 20 /
\plot 20 20  40 20 /
\setdashes <2pt>
\plot 40 0  40 20  /
\plot 40 40  60 40  60 20  /
\setsolid
\plot 36 32 44 28 /
\plot 56 12  64 8 /
\plot 46 2  54 -2 /

\multiput {$\bullet$} at 40 20  60 20  40 40  60 0  80 20  20 20
                                      40 0  /
\put {$\circ$} at 60 40 

\put {$\h{x}$} at 34 14
\put {$\h{y}$} at 68 14

\put {(9)} at 50 -34

\setplotsymbol ({\fiverm.})

\setplotsymbol ({$\cdot$})

\setdots <4pt>
\plot 120 20  140 20 /
\setsolid
\plot 100 20  120 20  120 40 /
\plot 160 0  140 0  140 20  160 20 /
\plot 100 0  120 0  120 -20 /
\setdashes <2pt>
\plot 40 0  40 20  /
\plot 40 40  60 40  60 20  /
\setsolid
\plot 36 32 44 28 /
\plot 56 12  64 8 /
\plot 46 2  54 -2 /

\multiput {$\bullet$} at 100 20  120 20  120 40  160 0  
                                      140 0  140 20  160 20  100 0  
                                      120 0  120 -20 / 
\multiput {$\circ$} at 140 40 140 -20  /

\put {$\h{x}$} at 114 14
\put {$\h{y}$} at 148 14

\put {A} at 110 30
\put {A} at 110 -10
\put {B} at 160 10

\put {(10)} at 130 -34

\setplotsymbol ({\scalebox{2.0}{.}})
\color{white}
\multiput {$\bullet$} at -80 30 140 -10 /
\color{black}
\put {(d)} at -110 20

\endpicture} at 10 30

\setcoordinatesystem units <1pt,1pt> point at -26 410

\put
{\beginpicture

\arrow <5pt> [.2,.67] from -80 0 to -80 20 
\arrow <5pt> [.2,.67] from -80 0 to -60 0 
\put {$\hj$} at -87 12
\put {$\hi$} at -72 -9

\setplotsymbol ({$\cdot$})

\setdots <4pt>
\plot -20 20  0 20 /
\setsolid
\plot -40 20  -20 20  -20 40 /
\plot  20 0  0 0  0 20  20 20 /
\plot -40 0  -20 0  -20 -20 /

\multiput {$\bullet$} at -40 20  -20 20  -20 40  20 0  
                                      0 0  0 20  20 20  -40 0  -20 0  
                                      -20 -20 / 
\multiput {$\circ$} at 0 40 0 -20  /

\put {$\h{x}$} at -26 14
\put {$\h{y}$} at 8 14

\put {A} at -30 30
\put {A} at -30 -10
\put {B} at 20 10

\put {(11)} at -10 -34

\setplotsymbol ({\fiverm.})
\arrow <5pt> [.2,.67] from  30 10 to 50 10 

\setplotsymbol ({$\cdot$})

\setdots <4pt>

\setsolid
\plot 60 20  80 20  80 40 /
\plot  140 0 120 0 120  20  140 20 /
\plot 60 0  80 0  80 -20 /
\setdashes <2pt>
\plot 80 40  100 40  120 40 120 20  /
\plot 80 20  100 20  100 0  120 0  /
\setsolid
\plot 76 32  84 28 /
\plot 116 12  124 8 /

\multiput {$\bullet$} at 60 20  80 20  80 40 140 0 120 0 
                                     120  20  140 20 60 0  80 0  80 -20 / 
\multiput {$\circ$} at 100 -20  100 0  100 20  100 40  120 -20  120 40  /

\put {$\h{x}$} at 74 14
\put {$\h{y}{+}\hi$} at 132 29

\put {A} at 70 30
\put {A} at 70 -10
\put {B} at 140 10

\put {(12)} at 100 -34

\setplotsymbol ({\scalebox{2.0}{.}})
\color{white}
\multiput {$\bullet$} at -80 30 140 -10 /
\color{black}
\put {(e)} at -110 20

\endpicture} at 10 30 

\normalcolor

\endpicture
\caption{\textit{The case analysis of concatenating two adsorbing polygons.
Either the edge $\alpha$, or the edge $\beta$ is present in case (5).
If $\alpha$ is present, then this case is handled similarly to case (4),
but $5$ edges are deleted and $3$ are added, increasing the length by
$2$.  If $\beta$ is present, then this case is analysed in (c).}}
\label{figure7}
\end{figure}

\vspace{2mm}
\noindent{\textit{(a) Height equal to one:}} The cases are shown in
figure \ref{figure7}(b).  These are found by noting that the vertex
marked $\times$ in figure \ref{figure7}(a)(2) cannot be occupied by 
either polygon.  However, the vertex marked by ${+}$ may be
either not occupied, giving figure \ref{figure7}(b)(3), or occupied,
in which case a case analysis give either figures \ref{figure7}(b)(4)
or \ref{figure7}(b)(5).

In the case of figure \ref{figure7}(b)(3), insert the edges
$\edge{(\h{x}{+}\hj)}{(\h{y}{+}\hj)}$,
$\edge{(\h{y}{+}\hj)}{(\h{y}{+}\hi{+}\hj)}$,
$\edge{(\h{x}{+}\hi{+}\hj)}{(\h{y}{+}\hi)}$ and
$\edge{\h{x}}{\h{y}}$, and then delete the edges
$\edge{\h{x}}{(\h{x}{+}\hj)}$ and
$\edge{\h{y}}{(\h{y}{+}\hi)}$.  This gives a polygon
of length $n{+}m{+2}$ and with $v{+}w$ visits.

The case in figure \ref{figure7}(b)(4) is similarly done.
Insert the edges
$\edge{(\h{x}{+}\hj)}{(\h{y}{+}\hj)}$,
$\edge{(\h{y}{+}\hj)}{(\h{y}{+}\hi{+}\hj)}$ and
$\edge{\h{x}}{\h{y}}$, and then delete the edges
$\edge{\h{x}}{(\h{x}{+}\hj)}$,
$\edge{\h{y}}{(\h{y}{+}\hi)}$ and
$\edge{(\h{y}{+}\hi)}{(\h{y}{+}\hi{+}\hj)}$.
This gives a polygon of length $n{+}m$ and with
$v{+}w$ visits.

The remaining case is figure \ref{figure7}(b)(5).
This case is more complicated, and is dealt with in
figure \ref{figure7}(c).  The first part of the construction
is to delete the subwalk of length seven steps containing
$\h{y}$ as shown in figure \ref{figure7}(c)(6) starting at
$\h{y}{+}3\hi$ and ending at
$\h{y}{-}\hj{+}3\hi$. This is followed by inserting the edge 
$\edge{(\h{y}{+}3\hi)}{(\h{y}{-}\hj{+}3\hi)}$.
The result is shown in figure \ref{figure7}{c}(7) where we
also add eight edges and delete two.  The newly added edges
are the subwalks 
$\h{x}{\sim}(\h{x}{-}\hj){\sim}(\h{x}{+}\hi{-}\hj){\sim}(\h{x}{+}2\hi{-}\hj)
{\sim}(\h{x}{+}2\hi){\sim}(\h{x}{+}3\hi){\sim}(\h{x}{+}2\hi{+}\hj)$
and
$(\h{x}{+}\hj){\sim}(\h{x}{+}\hi{+}\hj){\sim}(\h{x}{+}2\hi{+}\hj)$.
The two deleted edges are
$\edge{(\h{y}{+}\hi{+}\hj)}{\h{y}{+}2\hi{+}\hj)}$ and
$\edge{\h{x}}{(\h{x}{+}\hj)}$.
The resulting polygon has length $n{+}m$ and the total number 
of visits $v{+}w$ did not change.

This completes the cases when the vertices $\h{x}$ and $\h{y}$
are at height one.

\vspace{2mm}
\noindent{\textit{(b) Height larger than one:}}  The starting
point is again the configuration in figure \ref{figure7}(a)(2),
and we consider the cases that the vertex $\times$ is or is not
occupied (the vertex marked by ${+}$ will not play a role).

If the vertex marked by $\times$ in figure \ref{figure7}(a)(2)
is not occupied, then the concatenation is done as shown in 
figure \ref{figure7}(d)(8).  That is, insert edges
$\edge{(\h{x}{+}\hj)}{(\h{y}{+}\hj)}$,
$\edge{(\h{y}{+}\hj)}{\h{y}}$, 
$\edge{(\h{x}}{\h{x}{-}\hj}$, and
$\edge{(\h{x}{-}\hj)}{(\h{y}{-}\hj)}$, and then delete edges
$\edge{\h{x}}{(\h{x}{+}\hj)}$ and
$\edge{\h{y}}{(\h{y}{-}\hj)}$.
This gives a polygon of length $n{+}m{+}2$ with $v{+}w$
visits.

If $\times$ is occupied, then a case analysis shows that either
the conformations in figures \ref{figure7}(d)(9) or (d)(10) are
encountered. 

In the case of figure \ref{figure7}(d)(9), the concatenation is done as
shown. Insert edges
$\edge{(\h{x}{+}\hj)}{(\h{y}{+}\hj)}$,
$\edge{(\h{y}{+}\hj)}{\h{y}}$, and
$\edge{(\h{x}}{\h{x}{-}\hj}$, and then delete the edges
$\edge{\h{x}}{(\h{x}{+}\hj)}$,
$\edge{(\h{x}{-}\hj)}{(\h{y}{-}\hj)}$ and 
$\edge{(\h{y}}{\h{y}{-}\hj}$.
This gives a polygon of length $n{+}m$ with $v{+}w$ visits.

The remaining case is figure \ref{figure7}(d)(10), and the approach
here is shown in figure \ref{figure7}(e).

First notice that if the subwalks both marked by A in figure
\ref{figure7}(d)(10) are in two different polygons, then there is
a pair of parallel edges between the polygons and they can be concatenated without
changing length or the number of visits.

This, the subwalks in figure \ref{figure7}(d)(10) marked by A are
both in the first polygon, and the subwalk marked B is in the second
polygon. Proceed by translating the second polygon one step in the
$\hi$ direction as shown in figure \ref{figure7}(e)(11) to obtain
figure \ref{figure7}(e)(12).  This translation makes available the
vertices $\h{x}{+}\hi{+}N\hj$ for $N\in\{-2,-1,0,1\}$, while the
vertex $\h{x}{+}2\hi{+}\hj$ is also unoccupied.  The polygons
are now concatenated by inserting the subwalks
${\LA\right.}(\h{x}{+}\hj){\sim}(\h{x}{+}\hi{+}\hj){\sim}(\h{x}{+}2\hi{+}\hj)
{\sim}(\h{x}{+}2\hi){\left.\RA}$ and
${\LA\right.}\h{x}{\sim}(\h{x}{+}\hi){\sim}(\h{x}{+}\hi{-}\hj)
{\sim}(\h{x}{+}2\hi{-}\hj){\left.\RA}$, and deleting the edges
$\edge{\h{x}}{(\h{x}{+}\hj)}$ and
$\edge{(\h{y}{+}\hi)}{(\h{y}{+}\hi{-}\hj)}$.
The resulting polygon has length $n{+}m{+}4$ and has $v{+}w$ visits.

\begin{figure}[t]
\centering
\includegraphics[scale=0.4]{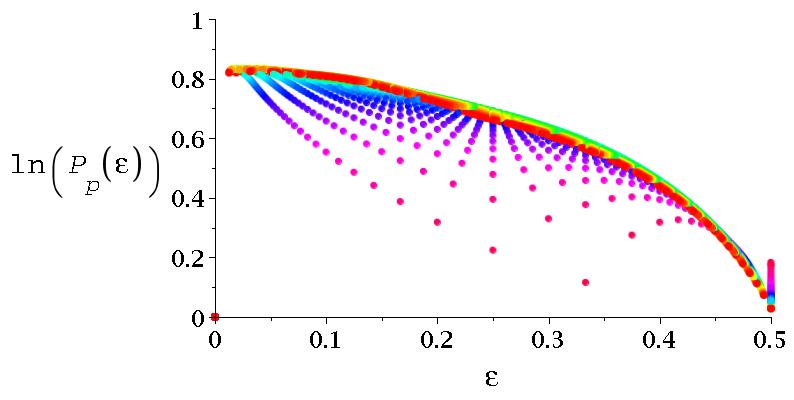}
\caption{Numerical estimate of $\log \C{P}_p (\eps)$ as a function of 
$\eps\in(0,1/2)$. The points have coordinates $(v/n,(\log c_n^+(v))/n)$ 
for $0\leq n \leq 160$ and $0\leq v \leq n/2$ and they accumulate on
$\log \C{P}_v(\eps)$ as $n$ increases, by theorem \ref{theorem3.6}.
Observe that the maximum in the data is well below $0.97$ for
small values of $\epsilon$, showing that for these values of $n$
(the length of the polygons) these estimates are still well below 
$\log \C{P}_p (\eps)$.  As $n\to \infty$ in this mode, the maximum
in the curve should approach $\log \mu_2 \approx 0.97$ at
$\eps=0$.}
\label{figure8}   
\end{figure}

\vspace{2mm}
\noindent{\textit{(c) Supermultiplicativity of adsorbing polygons:}}
While the number of visits to the adsorbing line is additive under the
concatenation of polygons in all the cases examined above, the resulting 
concatenated polygon may have lengths $n{+}m{+}k$ for $k\in\{0,2,4\}$
(and exactly $v{+}w$ visits).  This shows, in particular, that
\begin{equation}
p_n(v)\,p_m(w) \leq
p_{n+m}(v{+}w)
+ p_{n+m+2}(v{+}w)
+ p_{n+m+4}(v{+}w) .
\end{equation}
However, note that $p_n(v) \leq p_{n+2}(v)$, as seen by adding edges
at A in figure \ref{figure6}.  This simplifies the above to
$ p_n(v)\,p_m(w) \leq 3\, p_{n+m+4}(v{+}w)$, which, by replacing
$n \to n{-}4$ and $m \to m{-}4$, gives
\begin{equation}
p_{n-4}(v)\, p_{m-4}(w) \leq 3\, p_{n+m-4}(v{+}w) .
\end{equation}
Putting $q_n(v) = \sfrac{1}{3}\,p_{n-4}(v)$ then shows that 
$q_n(v)$ satisfies equation \Ref{9}.  Moreover $0 \leq q_n(v) \leq
p_n(v) \leq (2d)^n$ and $q_n(v) > 0$ if $2 \leq v \leq (n-4)/2$ and
$n\geq 8$ and $n$ even.  Thus, it follows from theorems \ref{theorem5}
and \ref{theorem3.6} that
\begin{equation}
\hspace{-2cm}
\log \C{P}_p(\eps)   
= \lim_{n\to\infty} \sfrac{1}{n} \log q_n(\lfl \eps n \rfl)
= \lim_{n\to\infty} \sfrac{1}{n} \log p_n(\lfl \eps n \rfl),
\;\hbox{for $\eps\in(0,\sfrac{1}{2})$}
\end{equation}
where the limits are taken through even $n$.

\subsubsection{Numerical estimation of $\log \C{P}_p(\eps)$:}
Adsorbing polygons were approximately enumerated using the
GARM algorithm \cite{RJvR08}, implemented with BFACF elementary
moves \cite{ACF83,BF81}.  Square lattice polygons were sampled of lengths
$n\leq 160$.  Overall, the numerical simulation in this model is more
difficult compared to adsorbing walks sampled in the last section, 
and the number of realised GARM sequences exceeded $1.5\times 10^8$.

The function $\log \C{P}_p(\eps)$ is approximated by
$\sfrac{1}{n} \log p_n(v/n)$ where $v=\lfl \eps n \rfl$ for large
values of $n$.  The results are shown in figure \ref{figure8}.
These data are not well approximated by functions similar to
equation \Ref{40} (but with $\epsilon$ replaced by $2\eps$).  However,
a non-linear fit using the function $a_0 + a_1 (1-2\eps)^{b_1}$ gives a good
fit,
\begin{equation}
\log \C{P}_p(\eps) \approx -0.0565 + 0.943\,(1-2\eps)^{0.408} .
\label{49}  
\end{equation} 
Taking the right derivative at $\eps=0$ gives the estimate of the
adsorption critical point for adsorbing square lattice polygons:
$\log a^+_p \approx 0.77$, although this estimate may have a large
uncertainly.  This estimate also suggests that the critical adsorption
point for walks, $a_c^+$ is strictly less than $a_p^+$ in the square
lattice (but it is known that $a_c^+ = a_p^+$ in the cubic lattice,
see theorem 9.23 in reference \cite{JvR15}).  Thus, we conjecture 
that $a_c^+ < a_p^+$ in the square lattice.

\begin{figure}[t!]
\beginpicture
\setcoordinatesystem units <1pt,1pt> 
\setlinear

\setplotarea x from -100 to 250, y from -20 to 35

\setcoordinatesystem units <1pt,1pt> point at -70 20

\put
{\beginpicture

\arrow <5pt> [.2,.67] from 10 25 to 10 40
\arrow <5pt> [.2,.67] from 95 45 to 95 60 
\circulararc 360 degrees from -30 4 center at -30 0 

\setplotsymbol ({$\cdot$})
\plot  -30 0  -40 0  -50 0  -50 10  -50 20  -40 20  -40 30  -30 30  -30 20 
-30 10  -20 10  -20 20  -10 20  -10 10  -10 0  0 0  10 0  10 10  10 20  /
\plot 60 0 60 10 60 20 70 20 70 30 70 40 80 40 80 30 80 20 90 20 100 20
100 30 90 30 90 40  100 40 110 40 120 40 120 30 130 30 130 20 120 20 
110 20 110 10 120 10 120 0 110 0 100 0  90 0  90 10 80 10 80 0 70 0 
60 0  /  
\multiput {\fns$\bullet$} at
-30 0  -40 0  -50 0  -50 10  -50 20  -40 20  -40 30  -30 30  -30 20 
-30 10  -20 10  -20 20  -10 20  -10 10  -10 0  0 0  10 0  10 10  10 20
60 0 60 10 60 20 70 20 70 30 70 40 80 40 80 30 80 20 90 20 100 20
100 30 90 30 90 40  100 40 110 40 120 40 120 30 130 30 130 20 120 20 
110 20 110 10 120 10 120 0 110 0 100 0  90 0  90 10 80 10 80 0 70 0 
  /
 
\put {$c^+_n (h)$} at -56 35 
\put {$p_n (h)$} at 146 35 
\multiput {$f$} at 17 40 102 55 /

\setplotsymbol ({\scalebox{2.0}{.}})

\plot -60 -4 20 -4 /
\plot 50 -4 140 -4 /
\plot 50 44 140 44 /

\endpicture} at 0 30 

\endpicture
\caption{\textit{Left: A self-avoiding walk from the origin pulled vertically by
a force $f$ at its endpoint.  Right:  A lattice polygon attached to the bottom
line pulled by a force $f$ in a plane through its highest vertices.}}
\label{figure9}
\end{figure}

\subsection{Pulled polygons in the square lattice}

\subsubsection{Existence of the density function:}
A ring polymer between two sticky plates pulled apart by a force $f$
can be modelled by a lattice polygon between two plates (and with
at least one vertex in each plate) which is being pulled apart by a force
$f$.  This model is illustrated on the right in figure \ref{figure9} where
the polygon is attached (has at least one vertex) in the bottom line
and in the top line, which can move as the height of the polygon changes.
A force pulling the top line stretches the polygon, as its highest vertices
are being pulled with the top line.

This model is quantified by letting $p_n(h)$ be the number of lattice polygons
in the half-square lattice, attached to the bottom line, and with highest 
vertices a height $h$ above the bottom line.  

A related model is shown on the left in figure \ref{figure9}, namely a 
positive self-avoiding walk from the origin in the half square lattice, pulled
at its endpoint by a vertical force $f$.  If $c_n^+(h)$ is the number of self-avoiding
walks from the origin in the half square lattice, of length $n$, and with endpoint
at height $h$, then the partition function of this model is given by
\begin{equation}
C_n^+(y) = \sum_{h=0}^n c_n^+(h)\, y^h ,
\end{equation}
where $y = e^{f/k_BT}$ with $f$ the applied vertical force ($T$ is the
temperature and $k_B$ the Boltzmann constant).  The free energy in this
model exists \cite{JOTW09} and is given by
\begin{equation}
\lambda(y) = \lim_{n\to\infty} \Sfrac{1}{n} \log C_n^+(y) .
\label{51}   
\end{equation}

\begin{figure}[t!]
\beginpicture
\setcoordinatesystem units <1pt,1pt> 
\setlinear

\setplotarea x from -110 to 250, y from -20 to 50

\setcoordinatesystem units <1pt,1pt> point at -70 20

\put
{\beginpicture

\arrow <5pt> [.2,.67] from 205 0 to 205 39
\arrow <5pt> [.2,.67] from 205 41 to 205 80

\setdashes <2pt>
\plot 50 40 200 40 /
\setsolid

\setplotsymbol ({$\cdot$})
\plot 60 0 60 10 60 20 70 20 70 30 70 40 80 40 80 30 80 20 90 20 100 20
100 30 90 30 90 40  100 40 110 40 120 40 120 30 130 30 130 20 120 20 
110 20 110 10 120 10 120 0 110 0 100 0  90 0  90 10 80 10 80 0 70 0 
60 0  /  
\plot 110 40 120 40 130 40 140 40 140 50 130 50 120 50 120 60 130 60 
140 60 150 60 150 50 160 50 160 40 170 40 170 50 180 50 180 60 190 60
190 70 180 70 180 80 170 80 170 70 170 60 160 60 160 70 160 80 
150 80 140 80 140 70 130 70 120 70 110 70 110 60 100 60 100 50
110 50 110 40 /
\multiput {\fns$\bullet$} at 
60 0 60 10 60 20 70 20 70 30 70 40 80 40 80 30 80 20 90 20 100 20
100 30 90 30 90 40  100 40 110 40 120 40 120 30 130 30 130 20 120 20 
110 20 110 10 120 10 120 0 110 0 100 0  90 0  90 10 80 10 80 0 70 0 
60 0 110 40 120 40 130 40 140 40 140 50 130 50 120 50 120 60 130 60 
140 60 150 60 150 50 160 50 160 40 170 40 170 50 180 50 180 60 190 60
190 70 180 70 180 80 170 80 170 70 170 60 160 60 160 70 160 80 
150 80 140 80 140 70 130 70 120 70 110 70 110 60 100 60 100 50
110 50 110 40 /
 
\put {$p_n (h)$} at 150 15 
\put {$p_m (k)$} at 80 65 

\put {$h$} at 210 20
\put {$k$} at 210 60

\plot 113 36 117 44 /

\setplotsymbol ({\scalebox{2.0}{.}})

\plot 50 -4 200 -4 /
\plot 50 84 200 84 /

\endpicture} at 0 30 

\endpicture
\caption{\textit{A polygon of length $n$ and height $h$ can be concatenated
with a polygon of length $m$ and height $k$ by placing the second polygon
so that its bottom-most, left-most edge is on top of the top-most, right-most
edge of the first polygon.  If these identified edges are deleted, then a polygon
of length $n{+}m{-}2$ and height $h{+}k$ is obtained.}}
\label{figure10}
\end{figure}

The partition function of lattice polygons between two plates and pulled
in the top plate is given by
\begin{equation}
R_n(y) = \sum_{h=1}^{(n-2)/2} p_n(h)\, y^h .
\end{equation}
The free energy in this model also exists, and is given by \cite{GJJW18}
\begin{equation}
\lambda(\sqrt{y}) = \lim_{n\to\infty} \Sfrac{1}{n} \log R_n (y) ,
\end{equation}
where $\lambda(y)$ is the free energy of pulled self-avoiding walks
(equation \Ref{51}).

By placing a polygon of length $n$ and height $h$ so that its top-most right-most
edge is identified with the bottom-most left-most edge of a polygon of
length $m$ and height $k$, and then deleting the two identified edges,
a polygon of length $n{+}m{-}2$ of height $h{+}k$ is obtained.  This is shown
in figure \ref{figure10}, and since there are $p_n(h)$ choices for the first 
polygon, and $p_m(k)$ choices for the second polygon,
$p_n(h)\,p_m(k) \leq p_{n+m-2}(h{+}k) $. Replacing $n$ by $n{+}2$ and
$m$ by $m{+}2$ to see that
\begin{equation}
p_{n+2}(h)\,p_{m+2}(k) \leq p_{n+m+2}(h{+}k) .
\end{equation}
That is, $p_{n+2}(h)$ satisfies equation \Ref{10}.
Since $p_{n+2}(h) \leq p_{n+2} \leq 4^{n+2}$ and, 
for a given $n{+}2$, $0\leq h \leq n/2$, if follows by theorems \ref{theorem5}
and \ref{theorem3.6} that the density function in this model is given by
\begin{equation}
\log \C{P}_h (\epsilon) 
= \lim_{n\to\infty} \sfrac{1}{n} \log p_n (\lfl \epsilon n \rfl),
\;\hbox{for $\eps\in(0,\sfrac{1}{2})$}.
\label{55}
\end{equation}

\begin{figure}[t]
\centering
\includegraphics[scale=0.4]{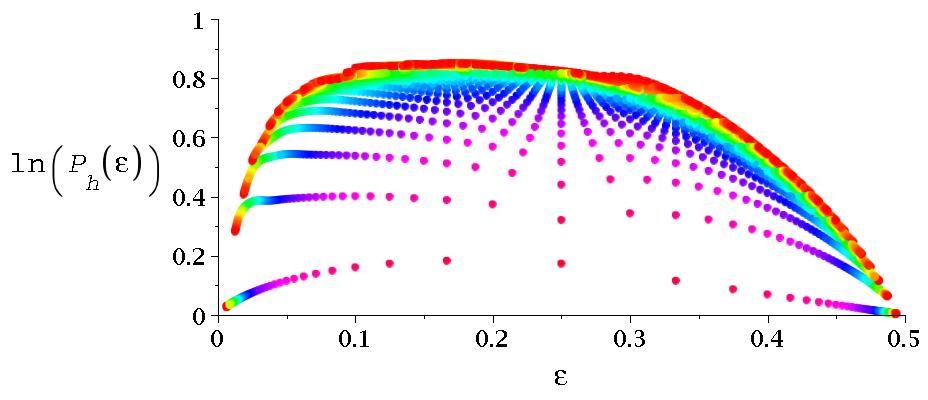}
\caption{Numerical estimate of $\log \C{P}_h (\eps)$ as a function of 
$\eps\in(0,1/2)$. The points have coordinates $(v/n,(\log p_n(h))/n)$ 
for $0\leq n \leq 160$ and $0\leq h \leq n/2$ and they accumulate on
$\log \C{P}_h(\eps)$ as $n$ increases, by theorem \ref{theorem3.6}.
Observe that the maximum in the data is below $0.97$, since the
finite length polygons are effectively confined to slits in the square lattice when
$\eps$ is fixed.  In the limit that $n$ goes to infinity the maximum in the
curve should approach $\log \mu_2 \approx 0.97$ as $\eps\to 0^+$.}
\label{figure11}   
\end{figure}

\subsubsection{Numerical estimation of $\log \C{P}_h(\eps)$:}
Pulled polygons were approximately enumerated using the
GARM algorithm \cite{RJvR08}, implemented with BFACF elementary
moves \cite{ACF83,BF81}.  Square lattice polygons were sampled of lengths
$n\leq 160$.  As for adsorbing polygons, the number of realised 
GARM sequences exceeded $9.6\times 10^7$.

The function $\log \C{P}_h(\eps)$ is approximated by
$\sfrac{1}{n} \log p_n(h/n)$ where $h=\lfl \eps n \rfl$ for large
values of $n$.  The results are shown in figure \ref{figure11}.
These data are not well approximated by functions similar to
equation \Ref{40} or equation \Ref{49}.  A non-linear fit using the 
function $a_1\, \eps^{b_0}\,(1-2\eps)^{b_1}$, giving
\begin{equation}
\log \C{P}_h(\eps) \approx 2.68\,\eps^{0.448}\,(1-2\eps)^{0.753} ,
\end{equation} 
was more successful.   The maximum in this approximation is
at $\epsilon \approx 0.187$ where $\log \C{P}_h(0.187) \approx 0.888$,
smaller than $\log \mu_2 \approx 0.9700$ (showing that for these
values of $n$ these estimates are below the limiting value $\log \C{P}_h(\eps)$). 

Using the model $a_1\, \eps^{b_0}\,(1-2\eps)^{b_1}+a_2\, \eps\,(1-2\eps)$
instead gives a slightly better approximation
\begin{equation}
\log \C{P}_h(\eps) \approx 5.18\,\eps^{0.586}\,(1-2\eps)^{0.799} 
- 3.79\, \eps\, (1-2\eps),
\end{equation} 
showing that there are still large uncertainties in the values of the exponents.
The maximum in this approximation is at $\epsilon \approx 0.177$ 
where $\log \C{P}_h(0.177) \approx 0.888$.  

Plotting these approximations
show that, while they adequately fit the data for $\eps\lesssim 0.1$ and
$\eps\gtrsim 0.3$, the data points for $\eps\in(0.1,0.3)$ are generally below
the approximations.

\def\P#1#2#3#4{\plot #1 #2 #3 #4 /}
\begin{figure}[b!]
\beginpicture
\setcoordinatesystem units <1pt,1pt> 
\setlinear

\setplotarea x from -100 to 250, y from -20 to 35

\setcoordinatesystem units <1.5pt,1.5pt> point at -70 20

\put
{\beginpicture

\circulararc 360 degrees from -50 5 center at -50 0 

\setplotsymbol ({$\cdot$})
\plot -50 0  -40 0  -40 10  -30 10  -30 20  -20 20  -20 10  -20 0  -10 0 
-10 -10 -20 -10 -20 -20 -10 -20 0 -20 10 -20 20 -20 20 -10 10 -10
0 -10 0 0 0 10 0 20 10 20  20 20  20 30  30 30  
30 20  40 20  40 10  30 10  30 0 40 0  50 0  50 10  60 10  70 10  70 0
70 -10 60 -10   /
\multiput {\fns$\bullet$} at
-50 0  -40 0  -40 10  -30 10  -30 20  -20 20  -20 10  -20 0  -10 0 
-10 -10 -20 -10 -20 -20 -10 -20 0 -20 10 -20 20 -20 20 -10 10 -10
0 -10 0 0 0 10 0 20 10 20  20 20  20 30  30 30  
30 20  40 20  40 10  30 10  30 0 40 0  50 0  50 10  60 10  70 10  70 0
70 -10 60 -10  /

\setplotsymbol ({.})

\setdashes <3.33pt>
\P{-30}{10}{-20}{10}
\P{-20}{0}{-20}{-10}
\P{-10}{-10}{-10}{-20}
\P{-10}{-10}{0}{-10}
\P{-10}{0}{0}{0}
\P{0}{-10}{0}{-20}
\P{10}{-10}{10}{-20}
\P{20}{20}{30}{20}
\P{30}{20}{30}{10}
\P{40}{10}{40}{0}
\P{40}{10}{50}{10}

\put {$c_n (k)$} at -20 30

\endpicture} at 0 30 

\endpicture
\caption{\textit{A self-interacting self-avoiding walk with eleven
nearest neighbour contacts.}}
\label{figure12}
\end{figure}

\subsection{Self-interacting self-avoiding walks in the square lattice}

A \textit{contact} in a self-avoiding walk is a pair of vertices which are
nearest neighbour in the lattice, but not nearest neighbour in the walk.
In figure \ref{figure12} a self-avoiding walk with eleven contacts is shown;
the contacts are indicated by a broken line between nearest neighbour
vertices.

If $c_n(k)$ is the number of self-avoiding walks in the square lattice,
of length $n$, and with $k$ contacts, then the density function of
this model is not known to exist.  The function $c_n(k)$ is not known to
satisfy a supermultiplicative, or submultiplicative, relation similar to equation
\Ref{10}.  A density function can be defined by the limit superior
\begin{equation}
\log \C{P}_k (\eps) = \limsup_{n\to\infty} \Sfrac{1}{n}
\log c_n(\lfl \eps n \rfl),
\;\hbox{for $\eps\in (0,1)$}.
\end{equation}
This function can be estimated numerically, although it is difficult
to collect data when $\eps \gtrsim 0.7$, since those states are
dense self-avoiding walks with ``frozen'' conformations which are
not easily updated by an algorithm.

\begin{figure}[t]
\centering
\includegraphics[scale=0.4]{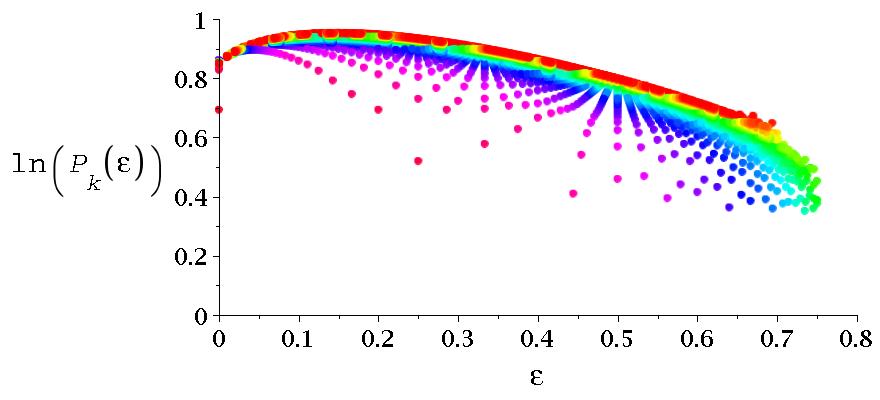}
\caption{Numerical estimate of $\log \C{P}_k (\eps)$ as a function of 
$\eps\in(0,0.8)$. The points have coordinates $(k/n,(\log c_n(k))/n)$ 
for $0\leq n \leq 100$ and they accumulate on $\log \C{P}_k(\eps)$ 
as $n$ increases. Observe that the maximum in the data is close to
$0.97$, where the function peaks at $\eps \approx 0.14$.}
\label{figure13}   
\end{figure}

In figure \ref{figure13} a numerical approximation is shown for
$\C{P}_k(\eps)$.  The data were collected using the GAS algorithm
\cite{RJvR08} implemented with Berretti-Sokal elementary moves 
\cite{BS85}. A total of $500$ sequences were sampled, each of length
$10^8$ iterations, for walks of lengths less than or equal to $100$.
Notice that no data were collected for $\eps\gtrapprox 0.75$ due to 
the difficulty of sampling data in the dense regime where walks have 
a large number of contacts.  

At $\eps=0$ the data accumulates on the point with value $0.86$,
and this gives the approximate value of the growth constant of 
neighbour-avoiding walks: $e^{0.86} \approx 2.364$.  In addition,
the data in figure \ref{figure13} suggest that the derivative
$\log \C{P}_k (\eps)$ is finite at $\eps=0$.  Thus,
the data were fitted to a quartic in $\eps$ with constant term equal 
to $0.86$.  This gives
\begin{equation}
\log \C{P}_k(\eps) \approx
0.86 + 1.38\,\nu - 6.26\,\nu^2 + 9.23\,\nu^3 - 5.46\,\nu^4 .
\end{equation}
The maximum in the data is seen at $\eps\approx 0.16$, which 
will be the natural density of contacts in self-avoiding walks.  At this
point $\log \C{P}_k(0.16) \approx 0.955$, close to the expected value
which is $\log \mu_2 \approx 0.9700$.  The apparent finite right derivative
at $\eps=0$ suggests a critical point $\log b_c \approx -1.38$ where
a (repulsive) force between nearest neighbour contacts becomes small
enough for the model to transition from a phase characterised by 
walks which are (mostly) nearest neighbour avoiding to a phase of
(expanded) self-avoiding walks.   A transition to the collapse
phase of self-avoiding walks should occur at a positive value 
of $\eps$ (corresponding to $\log b > 0$), where $\log \C{P}_k(\eps)$ 
will be non-analytic.

\section{Conclusions}

In this paper two-variable functions satisfying a supermultiplicative inequality 
and arising in the statistical mechanics of models of lattice clusters 
were examined.  In particular, it is shown that models with microcanonical
partition functions $p_n^\# (m)$ satisfying equation \Ref{7} have a 
(microcanonical) density function $\C{P}_\#(\eps)$.  Proofs of the existence
of the density function are given in theorems \ref{theorem5} and \ref{theorem3.6}
(where theorem \ref{theorem3.6} gives a more generalised limit
definition of $\C{P}_\#(\eps)$).  The proof of theorem \ref{theorem5} corrects 
the flawed proof of theorem 3.6 in reference \cite{JvR00}.  Existence of a 
thermodynamic limit in models satisfying equation \Ref{7} follows from 
theorem \ref{theorem6}. We have applied these results to various models 
in section \ref{examples}, and show that while $\log \C{P}_\#(\eps)$ arises
naturally in some models (adsorbing walks, and adsorbing and pulled polygons),
in other models (collapsing walks) existence of $\log \C{P}_\#(\eps)$ has 
not been showed (even if numerical data strongly suggest existence).

Generalisations of the results in this paper to multivariate functions in more than
two variables, or to models satisfying weaker supermultiplicative relations
than in equation \Ref{7}, remain unresolved.

\vspace{5mm}
\noindent{\bf Acknowledgements:} EJJvR acknowledges financial support 
from NSERC (Canada) in the form of Discovery Grant RGPIN-2019-06303.

\vspace{5mm}
\noindent{\bf References}
\bibliographystyle{plain}
\bibliography{density}

\end{document}